%% file: main.tex
\documentclass[a4paper,UKenglish,cleveref, autoref, thm-restate]{lipics-v2021}
\usepackage{comment}
\usepackage[utf8]{inputenc}
\newcommand{\quotes}[1]{``#1''}

\newcommand{\x}{X}

\usepackage{vlecomte}
\input{victor-macros}

\input{specific-defs}
\input{tikz-setup}

\nolinenumbers

\title{Hardness Amplification for Dynamic Binary Search Trees} %

\author{Shunhua	Jiang}{Columbia University, NY, USA}{	gingerbell00@gmail.com}{https://orcid.org/0000-0003-2226-7980}{}
\author{Victor	Lecomte}{Stanford University, CA, USA}{vc.lecomte@gmail.com}{https://orcid.org/0000-0002-6585-6980}{}
\author{Omri	Weinstein}{The Hebrew University of Jerusalem, Israel}{omri.weins@gmail.com}{https://orcid.org/0000-0002-9357-2299}{}
\author{Sorrachai	Yingchareonthawornchai}{The Hebrew University of Jerusalem, Israel}{sorrachai.cp@gmail.com}{https://orcid.org/0000-0002-7169-0163}{}

\authorrunning{S. Jiang, V. Lecomte, O. Weinstein, and S. Yingchareonthawornchai}

\Copyright {Shunhua Jiang, Victor Lecomte, Omri Weinstein and Sorrachai Yingchareonthawornchai} 

\acknowledgements{ This research was supported by the European Research Council (ERC) Starting grant CODY 101039914. This work was done in part while the authors (SJ, OW, SY) were visiting the Simons Institute for the Theory of Computing, UC Berkeley. }

\hideLIPIcs
\date{}
\begin{document}

\maketitle

\begin{abstract}

We prove direct-sum theorems for Wilber's two lower bounds [Wilber, FOCS'86] on the cost of access sequences in the binary search tree (BST) model. These bounds are central to the question of dynamic optimality~[Sleator and Tarjan, JACM'85]: the \emph{Alternation} bound is the only bound to have yielded online BST algorithms beating $\log n$ competitive ratio, while the \emph{Funnel} bound has repeatedly been conjectured to exactly characterize the cost of executing an access sequence using the optimal tree~[Wilber, FOCS'86, Kozma'16], and has been explicitly linked to splay trees~[Levy and Tarjan, SODA'19]. Previously,  the direct-sum theorem for the Alternation bound was known only when approximation was allowed~[Chalermsook, Chuzhoy and Saranurak, APPROX'20, ToC'24]. %
We use these direct-sum theorems to amplify the sequences from [Lecomte and Weinstein, ESA'20] that separate between Wilber’s Alternation and Funnel bounds, increasing the Alternation and Funnel bounds while optimally maintaining the separation.  
As a corollary, we show that Tango trees [Demaine et al., FOCS'04] are optimal among any BST algorithms that charge their costs to the Alternation bound. This is true for \emph{any} value of the Alternation bound, even values for which Tango trees achieve a competitive ratio of $o(\log \log n)$ instead of the default $O(\log \log n)$.  Previously,  the optimality of Tango trees was shown only for a limited range of Alternation bound~[Lecomte and Weinstein, ESA'20]. 

\end{abstract}


\section{Introduction}

Direct Sum theorems assert a lower bound on a certain complexity measure $\cC$ of a \emph{composed}\footnote{Formally speaking, direct-sum problems pertain to the complexity of solving $k$ separate copies of a problem $f$, rather than computing a composed function of $k$ copies $g(f(x_1),\ldots,f(x_k))$, but it is common to refer to both variations of the $k$-fold problem as direct-sums \cite{KRW95}.} problem $f\circ g$ in terms of the individual complexities of $f$ and $g$, ideally of the form 
$\cC(f\circ g) \approx \cC(f)+\cC(g)$. 
Direct Sums have a long history in complexity theory, as they provide a \emph{black-box} technique for amplifying the hardness of computational problems  $\cC(f^{\circ k}) \gtrsim k\cdot \cC(f)$, and are the most promising approach for proving several holy-grail lower bounds in complexity theory \cite{KRW95, GNW11, RM99, BBCR09, KLM021}. Moreover, a ``tensorization'' property of $\cC$ under composition allows to ``lift'' the problem and leverage its asymptotic behavior (e.g., concentration), which is not present in the single-copy problem --  this feature has been  demonstrated and exploited in various models, including combinatorial Discrepancy \cite{LSS08, Sherstov12}, Richness of data structure problems \cite{PT09}, decision trees \cite{RM99} and rank \cite{KLM021} to mention a few. 
Despite their powerful implications,  (strong) direct-sum scaling of composed problems are often simply false \cite{Raz11, Shitov18, Sha01_counterex_DP}, and highly depend on the underlying computational model. 

In this paper, we study direct sums in the \emph{online BST model}, motivated by the \emph{dynamic optimality} conjecture of Sleator and Tarjan \cite{ST85}. 
The dynamic optimality conjecture postulates the existence of 
an \emph{instance optimal} binary search tree algorithm (BST), namely, an online self-adjusting BST
whose running time\footnote{i.e. the number of pointer movements and tree rotations performed  by the BST} 
matches the best possible running time \emph{in hindsight} for any sufficiently long sequence of queries.
More formally, denoting by $\calT(\x)$ the operational time of a BST algorithm $\calT$ on an access sequence $\x= (x_1,\ldots,x_m) \in [n]^m$ 
of keys to be searched, the conjecture says that there is an online BST $\calT$ such that
$\forall\x$, $\calT(\x) \leq O(\OPT(\x))$, where $\OPT(\x) := \min_{\calT'} \calT'(\x)$ denotes the optimal 
offline cost for $X$.
In their seminal paper, Sleator and Tarjan \cite{ST85} conjectured that 
\emph{splay trees} are $O(1)$-competitive;  
A more recent competitor, 
the \emph{GreedyFuture} algorithm \cite{Luc88,DHIKP09,Munro00}, also forms a 
compelling candidate for constant-competitive dynamic optimality.  
However, the near-optimality of both Splay trees and GreedyFuture was proven only in special cases~\cite{Sundar92,Elmasry04,pettie2007splay,ChalermsookG0MS15,CGKMS15,GoyalG19,ChalermsookGJAP23,CPY24}, 
 and they are not known to be $o(\log n)$-competitive for general access sequences $\x$ 
(note that every balanced BST is trivially $O(\log n)$-competitive). After 35 years of active research, the best provable bound to date is an $O(\log\log n)$-competitive BST, starting with \emph{Tango trees}~\cite{DHIP07}, among other $O(\log\log n)$-competitive BST algorithms~\cite{BDDF10,Multisplay,ChalermsookCS20}. Interestingly, this progress was made possible due the development of \emph{lower bounds} in the BST model, as we discuss next. 

Indeed, a remarkable feature of the BST model---absent from general computational models (e.g., word-RAM)---is that it allows for nontrivial lower bounds on the search time of a \emph{fixed} query sequence $\x$: In general models, lower bounds against a specific input $\x$ do not make much sense as the best algorithm in hindsight can simply ``store and read-off the answer'' for $\x$.   
Nevertheless, in the BST model, even an all-knowing binary search tree must pay the cost of traversing the root-to-leaf path to retrieve keys. 
For example, there are classical examples of deterministic access sequences (e.g., \textit{bit-reversal} sequence~\cite{Wil89}) that require the worst case $\Omega(m\log n)$ total search time. This feature is what makes instance-optimality in the BST model an intriguing possibility. Our work focuses on two classic lower bounds due to Wilber~\cite{Wil89}, the Alternation and Funnel bounds (a.k.a, Wilber's first and second bounds), which are central to the aforementioned developments. 

\textbf{The Alternation and Funnel Bounds.} 
Essentially all BST lower bounds are derived from a natural geometric interpretation of the access sequence $X = (X_1,\ldots, X_m)$
as a point set on the plane, mapping the $i\th$ access $X_i$ to point $(X_i,i)$ (\cite{DHIKP09, IaconoSurvey}, 
see Figure~\ref{fig:geom-view}). The earliest lower bounds on $\OPT(\x)$ were proposed in an influential paper of Wilber \cite{Wil89}. The \emph{\IB} $\ALT_\calT(\x)$  
counts the total number of left/right alternations obtained by searching the keys $\x=(X_1,\ldots, X_m)$ 
on a \emph{fixed} (static) binary search tree $\calT$, where alternations are summed up over all nodes 
$v\in \calT$ of the ``reference tree'' $\calT$ (see Figure \ref{fig:alt} and the formal definition in Section~\ref{sec:prelims}). 
Thus, the Alternation bound is actually a family of lower bounds, optimized by the choice of the 
reference tree $\calT$, and we henceforth define $\ALT(\x) \coloneqq \max_\calT \ALT_\calT(\x)$.  The Alternation bound plays a key role in the design and analysis of Tango trees and their variants \cite{DHIP07, Multisplay}. In fact, 
\emph{all} non-trivial $o( \log n)$-competitive BST algorithms~\cite{DHIP07,BDDF10, Multisplay,ChalermsookCS20,CGJPY23} rely on the Alternation bound.  

In the same paper, Wilber proposed another lower bound---the \emph{Funnel Bound} $\Funnel(\x)$--- which is less intuitive and can be defined by the following process: %
Consider the geometric view $\{(X_i,i)\}_{i\in [m]}$ %
of the simple ``move-to-root'' algorithm that simply rotates each searched key $X_i$ to the root by a series of single rotations.
Then $\Funnel(X_i,i)$ is exactly the number of \emph{turns} on the path from the root to $X_i$ right before it is accessed \cite{AM78, IaconoSurvey}. The Funnel bound is then defined as $\Funnel(\x) := \sum_{i=1}^m \Funnel(X_i,i)$.
This view emphasizes the \emph{amortized} nature of the Funnel bound: at any point, there could be 
linearly many keys in the tree that are only \emph{one} turn away from the root, so one can only 
hope to achieve this bound in some amortized fashion.%

The Funnel bound has been repeatedly conjectured to tightly characterize the cost of an offline optimal algorithm~\cite{Wil89, kozma2016binary,ChalermsookCS20,lecomte2019settling}. Recently, Lecomte and Weinstein~\cite{lecomte2019settling} proved that the funnel bound is \textit{rotation-invariant}, meaning that the bound is preserved when the geometric representation of the input sequence is rotated by 90 degrees. This property also holds for an optimal algorithm~\cite{DHIKP09}, giving another evidence that the Funnel bound might give a tight characterization of the cost of an offline optimal algorithm.
While the Funnel bound does not have an algorithmic interpretation like $\Alt(X)$, Levy and Tarjan \cite{LT19} recently observed interesting  similarities between Splay trees and  the Funnel bound. The core difficulty in converting $\Funnel(X)$ into a BST algorithm  is its highly  \emph{amortized} nature (also a feature of Splay trees), compared to the Alternation bound which gives a point-wise lower bound on the retrieval time of $X_i\in \x$. As such, understanding the mathematical properties of the Funnel bound is important in its own right. %

\textbf{Access Sequence Composition and Known Direct-Sum Results.}    Informally, Direct Sum theorems assert a lower bound on the complexity measure of solving $R$ copies of a problem in a given computational model, in terms of the cost of solving a single copy, ideally $C(f^{\circ R}) \gtrsim \Omega(R)\cdot C(f)$, where $f^{\circ R}$ denotes certain $R$-copy \emph{composed} problem. Indeed, the precise notion of composition we use here (a-la  \cite{chalermsook2016landscape}) is crucial, as direct-sum theorems are subtle and often turn out to be \emph{false}~\cite{Raz11}.

A natural definition of sequence-composition in the BST model was introduced by~Chalermsook et al.~\cite{chalermsook2016landscape}. Let $X^{(1)}, \ldots, X^{(\ell)}$ be a sequence of $\ell$ access sequences where $X^{(i)} \in [n_i]^{m_i}$. That is, each sequence $X^{(i)}$ has $n_i$ keys and $m_i$ accesses where $m := \sum_i m_i, n := \sum_i n_i$. We view each sequence $X^{(i)}$ as the $i^{\text{th}}$ queue where we dequeue elements by the order of the sequence $X^{(i)}$ (i.e., in FIFO order). Let $\tld{X}  \in [\ell]^m$ be a sequence with keys in $[\ell]$ such that every $j \in [\ell]$ appears exactly $m_j$ times. We can view $\tld{X}$ as a \textit{template} which defines the ordering of dequeue operations among the $\ell$ queues. We define the \textit{composed sequence} $X :=  \tld{X}(X^{(1)}, \ldots, X^{(\ell)}) \in [n]^m$ as follows. For each $t = 1$ to $ m$, $X_t := q_t + \sum_{i < \tilde X_t}n_i$ where $q_t$ is the  next element dequeued from $X^{(\tilde X_t)}$.  We refer the reader to \Cref{sec:prelims} for the precise definition. 

The direct sum results for the optimal cost are well understood with applications to proving non-trivial bounds of binary search trees. In \cite{chalermsook2016landscape}, they prove (approximate) \textit{subadditivity} of the optimal cost on composed sequences. That is, 
\[
 \OPT(X) \leq   3\cdot \OPT\P{\tld{X}} + \sum_j \OPT\P{X^{(j)}}.
\]
The subadditivity of optimal cost finds application in proving the linear optimal bounds for \quotes{grid} sequences, and a strong separation in the hierarchy of \textit{lazy finger bounds}~\cite{chalermsook2016landscape}.  On the other hand, \cite{ChalermsookCS20} recently proved \textit{superadditivity} of the optimal cost on composed sequences. That is, 
\begin{equation} \label{eq:superadditivity of opt}
 \OPT(X) \geq    \OPT\P{\tld{X}} + \sum_j \OPT\P{X^{(j)}}.
\end{equation}
The superadditivity of optimal cost finds an application in designing a new $O(\log \log n)$-competitive online BST algorithm based on purely geometric formulation~\cite{ChalermsookCS20}. 

However, the direct sum results for Wilber's bounds are poorly understood despite their importance to the pursuit of dynamic optimality.  The only published work that we are aware of is the approximate subadditivity of the Alternation bound~\cite{ChalermsookCS20}. That is,  they proved that 
\begin{equation} \label{eq:approx subadditivity}
    \Alt(X)  \leq 4\cdot \Alt\P{\tld{X}} + 8\cdot \sum_j \Alt\P{X^{(j)}} + O(|X|).
\end{equation}
Their proof is quite involved, and it is based on geometric arguments and the probabilistic method. This finds applications in proving the separation between the Alternation and Funnel bounds. In \cite{ChalermsookCS20}, they used approximate subadditivity of the Alternation bound (\Cref{eq:approx subadditivity}) to prove a near-optimal separation between the Alternation and Funnel bounds. That is, they constructed a sequence $Y$ such that the gap between  $\Alt(Y)$ and $\Funnel(Y)$ is as large as $\Omega(\frac{\log\log n}{\log\log\log n})$. This gap is nearly optimal because the upper bound of Tango tree~\cite{DHIP07} implies that the gap must be $O(\log \log n)$. This gap has been closed by an independent work by Lecomte and Weinstein~\cite{lecomte2019settling}, proving the optimal separation between the Alternation and Funnel bounds. That is, they constructed an instance $Y$ such that the gap between  $\Alt(Y)$ and $\Funnel(Y)$ is as large as $\Omega(\log\log n)$. 

Furthermore, \cite{ChalermsookCS20} also used the approximate subadditivity of the Alternation bound (\Cref{eq:approx subadditivity}) to prove the $\Omega(\frac{\log\log n}{\log\log\log n})$ gap between $\Alt(Y)$ and $\cGB(Y)$ where $\cGB(Y)$ denotes the \textit{Consistent Guillotine Bound} (cGB), a lower bound measure that is an extension of $\Alt(Y)$.

\subsection{Our Results}

We prove that the Alternation bound is subadditive whereas the Funnel bound is superadditive for composed sequences. More precisely, we prove the following theorem. 
\begin{theorem} [Direct-Sum Theorem for Wilber's Bounds] \label{thm:direct sums}
     Let $X \ce \tld{X}(X^{(1)}, \ldots, X^{(l)})$ be a composed sequence. Then
\begin{itemize}  
    \item      $\Alt(X)  \leq \Alt\P{\tld{X}} + \sum_j \Alt\P{X^{(j)}} +  O(|X|)$, and
    \item     $\Funnel(X) \geq \Funnel\P{\tld{X}} + \sum_j \Funnel\P{X^{(j)}} - O(|X|).$ 
\end{itemize}
\end{theorem}

 Our proof of the subadditivity of the Alternation bound is simpler and yields stronger bounds than the proof in~\cite{ChalermsookCS20} (Theorem 3.6 in their arXiv version).  Direct sum theorems are a natural black-box tool for hardness amplification, as they effectively reduce complex lower bounds to a simpler ``one-dimensional'' problem.  Indeed, as a showcase application, we use the base-case separation proved in \cite{lecomte2019settling} along with \Cref{thm:direct sums} to amplify both Wilber's bounds. %
 Let $\ol \ALT(X) := \ALT(X)/|X|,$ and $\ol \Funnel(X) := \Funnel(X)/|X|$. They proved that there is a sequence $Y$ such that 
 \[ 
\al{
\ol\Alt\P{Y} &\le O(1), \mbox{but }\\
\ol\Funnel\P{Y} &\ge \Omega\P{\log\log n }.} \] 
Note that the sequence $Y$ is \textit{easy} w.r.t. the Alternation bound since $\ol\Alt\P{Y} \le O(1)$. We use the sequence $Y$ as a base-case and apply \Cref{thm:direct sums} to construct \textit{hard} sequence w.r.t. the funnel bound while maintaining the separation as in the following theorem \footnote{This can be viewed as hardness amplification because the new sequence becomes harder from the optimum's point of view without losing the gap too much.}.

\begin{theorem}[Hardness Amplification]\label{thm:main_boosting}
There is a constant $K>0$ such that for any $n$ of the form $2^{2^r}$ and any power-of-two $R \le \f{\log n}{K}$, there is a sequence $Y_n^{\circ R} \in [n]^{m'}$ with $m' \le \poly(n)$ such that
\[ 
\al{
\ol\Alt\P{Y_n^{\circ R}} &\le O(R)\\
\ol\Funnel\P{Y_n^{\circ R}} &\ge \Omega\P{R\log \P{\f{\log n}{R}}}.} \] 
\end{theorem}

\textbf{Remark 1.} We emphasize that the approximate subadditivity of the Alternation bound (\Cref{eq:approx subadditivity}) is not sufficient for such hardness amplification. On the other hand, one could also use the superadditivity of the optimal cost (\Cref{eq:superadditivity of opt}) instead of the Funnel bound to prove hardness amplification.

\textbf{Tightness of the Separation.} As a corollary of \Cref{thm:main_boosting}, we can derive the following trade-offs between the multiplicative and additive factors for the Alternation bound.

\begin{theorem}\label{thm:lower bound alt charging}

    Let $\alpha,\beta: \N \to \R_{\ge 1}$ be any functions such that some BST algorithm achieves an amortized cost of $\alpha(n)\ol\Alt(X) + \beta(n)$ for all access sequences $X$ over $n$ keys. Then $\alpha(n) \ge \Omega\P{\log \P{\f{\log n}{\beta(n)}}}$.

\end{theorem}

As we discuss below, the trade-offs are tight with the matching upper bounds from Tango Trees (which can be derived directly from Tango trees). For convenience, we present self-contained BST algorithms with the matching upper bounds.%

\begin{lemma} \label{lem:bootstrapping tango} There is a BST algorithm that takes an integer $k \leq \log n$ as a parameter and serve the sequence $X = (X_1, \cdots, X_m)$  with the total access time of
\[
O\P{ \P{\ALT(X) + m \cdot \frac{\log n}{k}} \cdot (\log k + 1)}.
\]
\end{lemma}
We prove \Cref{lem:bootstrapping tango} in \Cref{sec:appendix:boosting tango}. 

In this algorithm, the additive cost is $\Theta\P{\f{\log n \log k}{k}}$ and multiplicative cost $\Theta(\log k)$. By \Cref{thm:lower bound alt charging}, if $\beta(n) = \Theta\P{\f{\log n \log k}{k}}$, then  we have
\[
\alpha(n)  \ge \Omega\P{\log \P{\f{\log n}{\beta(n)}}} = \Omega\P{\log \P{\f{k}{\log k}}} = \Omega(\log k),
\]
so our trade-off is optimal for any sufficiently large $k \le \log n$.

\textbf{Optimality of Tango Trees.}  As another corollary of \Cref{thm:main_boosting}, we prove the optimality of Tango Trees among any algorithm charging its cost to Wilber's Alternation bound for all values of the $\ol \ALT(X)$. Note that $\ol \ALT(X) \leq O(\log n)$. Previously, the optimality of Tango trees is known only when $\ol \ALT(X) = O(1)$~\cite{lecomte2019settling}. 

The basic idea of Tango Trees is to ``mimic'' Wilber's alternation bound via a BST, by dynamically maintaining a partition  
of the reference tree $\calT$ into \emph{disjoint paths}, formed by designating, for each node $x \in \calT$, the unique ``preferred'' descendant in $\calT$ (left or right) which was \emph{acessed most recently}. Since each ``preferred path'' has length $|p|\leq \text{depth}(\calT)\leq \log n$, every \emph{path} can itself be stored in a BST, so assuming these paths can be dynamically maintained (under split and joins), searching for the predecessor of a key $x_i$ \emph{inside} each path only takes $O(\log \log n)$ time, until the search ``falls-off'' the current preferred path and switches to a different one. The key observation is that this ``switch'' can be charged to $\ALT_\calT(\x)$, as it certifies a new alternation in Wilber's lower bound, hence OPT must pay for this move as well. This elegant argument directly leads to the aforementioned $O(\log\log n)\cdot \OPT(\x)$ search time.  

The analysis of Tango trees relies on charging the algorithm's cost to the Alternation bound. One may ask if the bound can be improved using a clever algorithm (not necessarily Tango trees) so that we can charge $o(\log \log n)$ factor to the Alternation bound. 
     Unfortunately, the answer is no as there are known examples of access sequences $\tilde{\x}$ with $\ALT(\tilde{\x}) = O(m)$ but $\OPT(\tilde \x) = \Theta(m\log \log n)$ \cite{lecomte2019settling,ChalermsookCS20}. In light of this,  Tango trees are indeed off by a factor $\Theta(\log\log n)$ from $\ALT(\tilde{\x})$.
%
Interestingly, when $\OPT(\x) \gtrsim m \frac{\log n}{2^{o(\log \log n)}}$, one can do somewhat better. Let $\ol \OPT(X) := \OPT(X)/|X|$. The Tango tree, presented by \cite{DHIP07} (see the discussion in their Section 1.5), has a competitive ratio of


\begin{equation} \label{eq:fine-grained tango}
     O\P{1+ \log \frac{\log n}{\ol \OPT(\x)}} = o(\log \log n).    
\end{equation}

 
The condition that allows $o(\log\log n)$-competitiveness is rather narrow: the amortized optimal cost $\ol \OPT(X)$ must be very close to $\log n$ to achieve $o(\log \log n)$-competitiveness. Can we achieve $o(\log\log n)$-competitiveness with a wider range of $\ol\OPT(X)$ using $\ALT(X)$?  %

Unfortunately, the answer is no. More generally, we prove a matching lower bound: the competitive ratio of any BST algorithm based on the Alternation bound must be at least $\Omega(\log \frac{\log n}{\ol \OPT(X)})$, matching to the upper bound of the Tango tree by~\Cref{eq:fine-grained tango}. More precisely, we prove the following theorem whose proof is a small modification of the proof in \Cref{thm:lower bound alt charging}. %

\begin{theorem} \label{thm:lower bound alt as function of opt}
    Let $\alpha: \N\times\R_{\ge 1} \to \R_{\ge 1}$ be any function such that some BST algorithm achieves an amortized cost of $\alpha\P{n,\ol\OPT(X)} \cdot \P{\ol\Alt(X)+1}$ for all access sequences $X$ over $n$ keys. Then $\alpha(n,s) \ge \Omega\P{\log \P{\f{\log n}{s}}}$.
\end{theorem}

 As a corollary of \Cref{thm:lower bound alt as function of opt},\footnote{We remark that we cannot set $\beta(n) = \ol \OPT(X)$ in \Cref{thm:lower bound alt charging} because  $\beta(n)$ does not depend on $X$ and, even if $\ol \OPT(X)$ is a function of $n$, the construction of the sequence $X$ depends on $\beta(n)$. We need the lower bound of \Cref{thm:lower bound alt as function of opt}.} the lower bound of $\Omega(\log \frac{\log n}{\ol \OPT(X)})$ follows by setting $s$ to be within constant factor of $\ol \OPT(X)$.  This holds for every BST algorithm based on the Alternation bound.  With the matching upper bound by~\Cref{eq:fine-grained tango},  Tango tree optimally utilizes the Alternation bound.

\subsection{Further Related Work}


  Splay trees and GreedyFuture are prime candidates for the dynamic optimality conjecture since they both satisfy many important properties of dynamic trees including static optimality~\cite{ST85}, working-set property~\cite{ST85,Fox11},  dynamic finger property~\cite{Cole00,ChalermsookJ20}, and more~(see the surveys~\cite{IaconoSurvey,kozma2016binary} for an overview of the results in the field). Although they are not yet known to have $o(\log n)$-competitiveness, they have substantially better bounds for special cases. For example, both Splay trees and GreedyFuture are dynamically optimal for \textit{sequential access} sequences~\cite{Elmasry04,Fox11}. For   \textit{deque sequences}, Splay trees are $O(\alpha^*(n))$-competitive~\cite{pettie2007splay} whereas GreedyFuture is $O(\alpha(n))$-competitive~\cite{ChalermsookGJAP23}. Here, $\alpha(n)$ denotes the inverse Ackermann function and $\alpha^*(n)$ is the iterated function of $\alpha(n)$. The sequential and deque sequences are special cases of the \textit{pattern-avoiding access} sequences~\cite{CGKMS15}. For any fixed-size pattern, GreedyFuture is  $O(2^{(1+o(1))\alpha(n)})$-competitive for pattern-avoiding access sequences~\cite{CPY24}. It was shown recently that an optimal BST algorithm takes $O(n)$ total cost for any fixed pattern~\cite{OptPattern23} and there is an optimal algorithm for sorting pattern-avoiding sequences~\cite{abs-2409-07868}.  The bounds for specific classes of patterns can be improved if preprocessing is allowed~\cite{CGKMS15,GoyalG19} (i.e., the initial tree can be set before executing the online search sequences).  Recently, a slight modification of GreedyFuture was shown to be $O(\sqrt{\log n})$-competitive~\cite{CGJPY23}.  An important application of GreedyFuture (or any competitive online BSTs) includes adaptive sorting using treesort~\cite{BlellochD23,CPY24} and heapsort~\cite{KS18}. %

 The lower bounds in the literature other than Wilber's bounds include the \textit{maximum independent rectangle}  and \textit{SignedGreedy} bounds~\cite{DHIKP09}, which subsume Wilber's Alternation and Funnel bounds. A similar lower bound was presented by Derryberry et al.~\cite{DerrberrySW05}.    Recently, \textit{Guillotine} Bound~\cite{ChalermsookCS20} was introduced, which is a generalization of Wilber's Alternation bound. Unfortunately, it is unclear how to design an algorithm that utilizes these bounds.  Recently, Sadeh and Kaplan~\cite{SadehK23} proved that the competitive ratio of GreedyFuture cannot be less than 2 for the multiplicative factor, or $o(m\log\log n)$ for the additive factor.

\subsection{Paper Organization}

We first describe terminologies and notations in \Cref{sec:prelims}. We prove the direct-sum results for Wilber's bounds (\Cref{thm:direct sums}) in \Cref{sec:composition}. In \Cref{sec:boosting}, we prove the hardness amplification of Wilber's bounds while maintaining their separation (\Cref{thm:main_boosting}) and we also prove \Cref{thm:lower bound alt charging}. We prove the superadditivity of the Funnel Bound in \Cref{sec:appendix:supperadditivity funnel}.  

\input{prelims}
\input{effect-of-composition}
\input{boosting-separation}



\bibliographystyle{plainurl}

\appendix

\input{appendix}
\end{document}

%% file: victor-macros.tex
\newcommand{\al}[1]{\begin{aligned}#1\end{aligned}}

\newcommand{\ub}{\underbrace}
\newcommand{\ol}{\overline}
\newcommand{\tld}{\widetilde}

\newcommand{\f}{\frac}
\newcommand{\sse}{\subseteq}
\newcommand{\ce}{\coloneqq}

\DeclarePairedDelimiter{\set}{\{}{\}}
\newcommand{\SET}[1]{\set*{#1}}
\DeclarePairedDelimiter{\p}{(}{)}
\renewcommand{\P}[1]{\p*{#1}}

\DeclarePairedDelimiter{\abs}{|}{|}
\newcommand{\ABS}[1]{\abs*{#1}}

\newcommand{\INTERNALCOND}[2]{#1\,\middle|\,#2}

\newcommand{\SETCO}[2]{\!\SET{\INTERNALCOND{#1}{#2}}}

\newcommand{\poly}{\mathrm{poly}}

\newcommand{\N}{\mathbb{N}}

\newcommand{\R}{\mathbb{R}}

\renewcommand{\and}{\wedge}

\newcommand{\union}{\cup}

\DeclareMathOperator{\W}{\mathbf{W}}

\renewcommand{\th}{^{\mathrm{th}}}

\newcommand{\calT}{\mathcal{T}}


%% file: specific-defs.tex
\newcommand{\Tout}{\calT_\mathrm{out}}

\newcommand{\switches}{\#\mathrm{switches}}
\newcommand{\mix}{\mathrm{mix}}
\newcommand{\mixValue}{\mathrm{mixValue}}
\newcommand{\ttL}{\mathtt{L}}
\newcommand{\ttR}{\mathtt{R}}
\newcommand{\TL}{{\cT_\ttL}}
\newcommand{\TR}{{\cT_\ttR}}

\newcommand{\PL}{{P_\ttL}}
\newcommand{\PR}{{P_\ttR}}
\newcommand{\FL}{{F_\ttL(P,p)}}
\newcommand{\FR}{{F_\ttR(P,p)}}

\newcommand{\OPT}{\mathsf{OPT}}
\newcommand{\cGB}{\mathsf{cGB}}

\newcommand{\ALT}{\mathsf{Alt}}
\newcommand{\Alt}{\ALT}
\renewcommand{\W}{\mathsf{Funnel}}
\newcommand{\Funnel}{\W}

\newcommand{\spl}{\mathrm{split}}

\newcommand{\rect}[1]{\square #1}

\newcommand{\cC}{\mathcal{C}}

\newcommand{\cT}{\mathcal{T}}

\newcommand{\FB}{funnel bound}
\newcommand{\IB}{alternation bound}

%% file: tikz-setup.tex
\usetikzlibrary{angles,quotes,decorations.pathreplacing,calc,decorations.markings,shapes.misc,positioning,patterns}
\tikzstyle{dot} = [circle, fill, minimum size=.1cm, inner sep=0]
\tikzset{cross/.style={cross out, draw, 
         minimum size=2*(#1-\pgflinewidth), 
         inner sep=0pt, outer sep=0pt},
         cross/.default={.9mm}}
\tikzset{every picture/.style={thick, font=\small, scale=0.8}}


%% file: prelims.tex
\section{Preliminaries}\label{sec:prelims}

We follow notations and terminologies from~\cite{lecomte2019settling}. 

\begin{definition}[access sequence]
    An access sequence is a finite sequence $X = (X_1, \ldots, X_m) \in S^m$ of values from a finite set of keys $S \sse \R$. Usually, we let $S \ce [n]$.
\end{definition}

To make our definitions and proofs easier, we will work directly in the geometric representation of access sequences as (finite) sets of points in the plane $\R^2$.
\begin{definition}[geometric view]
Any access sequence $X=(X_1, \ldots, X_m) \in S^m$ can be represented as the set of points $G_X\ce\{(X_t, t) \mid t \in [m]\}$, where the $x$-axis represents the key and the $y$-axis represents time (see Figure~\ref{fig:geom-view}).
\end{definition}

\input{figs/geom-view}

By construction, in $G_X$, no two points share the same $y$-coordinate. We will say such a set has ``distinct $y$-coordinates''. In addition, we note that it is fine to restrict our attention to sequences $X$ without repeated values.\footnote{Indeed, Appendix~E in \cite{CGKMS15} gives a simple operation that transforms any sequence $X$ into a sequence $\spl(X)$ without repeats such that $\OPT(\spl(X)) = \Theta(\OPT(X))$. Thus if we found a tight lower bound $L(X)$ for sequences without repeats, a tight lower bound for general $X$ could be obtained as $L(\spl(X))$.}
The geometric view $G_X$ of such sequences also has no two points with the same $x$-coordinate. We will say that such a set has ``distinct $x$- and $y$-coordinates''.


\begin{definition}[$x$- and $y$-coordinates]
For a point $p \in \R^2$, we will denote its $x$- and $y$-coordinates as $p.x$ and $p.y$. Similarly, we define $P.x \ce \{p.x \mid p \in P\}$ and $P.y \ce \{p.y \mid p \in P\}$.
\end{definition}

We start by defining the \emph{mixing value} of two sets: a notion of how many two sets of numbers are interleaved. It will be useful in defining both the \IB{} and the \FB{}. We define it in a few steps.
\begin{definition}[mixing string]
Given two disjoint finite sets of real numbers $L,R$, let $\mix(L,R)$ be the string in $\{\ttL,\ttR\}^*$ that is obtained by taking the union $L \cup R$ in increasing order and replacing each element from $L$ by $\ttL$ and each element from $R$ by $\ttR$. For example, $\mix(\{2,3,8\},\{1,5\}) = \mathtt{RLLRL}$.
\end{definition}

\begin{definition}[number of switches]
Given a string $s \in \{\ttL,\ttR\}^*$, we define $\switches(s)$ as the number of side switches in $s$. Formally,
\[\switches(s) \coloneqq 
\#\{t \mid s_t \neq s_{t+1}\}.\]
For example,
$\switches(\mathtt{LLLRLL})=2$. Note that if we insert characters into $s$, $\switches(s)$ can only increase.
\end{definition}

\begin{definition}[mixing value]
Let $\mixValue(L,R)\coloneqq \switches(\mix(L,R))$ (see Figure~\ref{fig:mixvalue}).  
\end{definition}

\input{figs/mixvalue.tex}

The mixing value has some convenient properties, which we will use later:
\begin{fact}[properties of $\mixValue$]\label{fact:props-mixvalue}
Function $\mixValue(L,R)$ is:
\begin{enumerate}[(a)]
\item symmetric: $\mixValue(L,R) = \mixValue(R,L)$;
\item monotone: if $L_1 \subseteq L_2$ and $R_1 \subseteq R_2$, then $\mixValue(L_1,R_1) \leq \mixValue(L_2,R_2)$;
\item superadditive under concatenation: if $L_1,R_1 \subseteq (-\infty,x]$ and $L_2,R_2 \subseteq [x,+\infty)$, then $\mixValue(L_1 \cup L_2, R_1 \cup R_2) \geq \mixValue(L_1,R_1) + \mixValue(L_2,R_2)$.
\end{enumerate}
Finally, $\mixValue(L,R) \leq 2\cdot\min(|L|,|R|)+1$.
\end{fact}

The definitions of Wilber's Alternation and Funnel bounds ($\Alt(X), \Funnel(X)$) are standard in the literature.

We now give precise definitions of Wilber's two bounds.\footnote{These definitions may differ by a constant factor or an additive $\pm O(m)$ from the definitions the reader has seen before.
We will ignore such differences, because the cost of a BST also varies by $\pm O(m)$ depending on the definition, and the interesting regime is when $\OPT(X) = \omega(m)$.}

\input{figs/alt.tex}

\begin{definition}[\IB{}]\label{def:ib}
Let $P$ be a point set with distinct $y$-coordinates, and let $\cT$ be a binary search tree over the values $P.x$.
We define $\ALT_\cT(P)$ using the recursive structure of $\cT$. If $\cT$ is a single node, let $\ALT_\cT(P)\coloneqq 0$. Otherwise, let $\TL$ and $\TR$ be the left and right subtrees at the root. Partition $P$ into two sets $\PL\coloneqq \{p \in P \mid p.x \in \TL\}$ and $\PR \coloneqq \{p \in P \mid p.x \in \TR\}$ and consider the quantity
$\mixValue(\PL.y, \PR.y)$, which
describes how much $\PL$ and $\PR$ are interleaved with time (we call each switch between $\PL$ and $\PR$ a ``preferred child alternation''). Then
\begin{equation}\label{eq:ib}
\ALT_\cT(P) \coloneqq \mixValue(\PL.y, \PR.y) + \ALT_\TL(\PL) + \ALT_\TR(\PR).
\end{equation}
The alternation bound is then defined as the maximum over all trees:
$$
\ALT(P)\ce \max_{\cT} \ALT_{\cT}(P).
$$
In addition, for an access sequence $X$, let $\ALT_\cT(X) \coloneqq \ALT_\cT(G_X)$ and $\ALT(X) \coloneqq \ALT(G_X)$.
\end{definition}

For the next few definitions, it is helpful to refer back to Figure~\ref{fig:fb}. In particular, $\FL$ and $\FR$ (the left and right funnel) correspond to the points marked with $\ttL$ and $\ttR$.

\input{figs/fb}

\begin{definition}[left and right funnel]\label{def:lr-funnel}
Let $P$ be a point set. For each $p \in P$, we say that access $q \in P$ is in the left (resp. right) funnel of $p$ within $P$ if $q$ is to the lower left (resp. lower right) of $p$ and $\rect{pq}$ is empty. Formally, let
\[\FL \coloneqq \{q \in P \mid q.y < p.y \,\wedge\, q.x < p.x \,\wedge\, P \cap \rect{pq} = \{p,q\}\}\]
and
\[\FR \coloneqq \{q \in P \mid q.y < p.y \,\wedge\, q.x > p.x \,\wedge\, P \cap \rect{pq} = \{p,q\}\}.\]
We will collectively call $\FL \cup \FR$ the funnel of $p$ within $P$.
\end{definition}

\begin{definition}[\FB{}]\label{def:fb}
Let $P$ be a point set with distinct $y$-coordinates. For each $p \in P$, define quantity
\[f(P,p) \coloneqq \mixValue(\FL.y, \FR.y),\]
which describes how much the left and right funnel of $p$ are interleaved in time.

Then
\[\W(P) \coloneqq \sum_{p \in P} f(P,p).\]
In addition, for an access sequence $X$, let $\W(X) \coloneqq \W(G_X)$.
\end{definition}

\begin{definition}[amortized bounds]
For any sequence $X \in S^m$, define amortized versions of the optimal cost and the Wilber bounds:
$$
\al{
\ol{\OPT}(X) \ce \f{\OPT(X)}{m}, \quad
\ol{\Alt}(X) \ce \f{\Alt(X)}{m}, \quad
\ol{\Funnel}(X) \ce \f{\Funnel(X)}{m}.
}
$$
\end{definition}

\begin{definition}[composed sequence, see \cite{chalermsook2016landscape}]\label{def:composition}
Let $S_1, \ldots, S_l$ be disjoint sets of keys with increasing values (i.e. $\forall x \in S_{j}, x' \in S_{j+1}$, we have $x < x'$).
For each $j \in [l]$, let $X^{(j)} \in S_j^{m_j}$ be an access sequence with keys in $S_j$, and
let $\tld{X}$ be a sequence with keys in $[l]$ such that every $j \in [l]$ appears exactly $m_j$ times (its total length is $m \ce m_1 + \cdots + m_l$).
Then we define the composed sequence
$$
X = \tld{X}(X^{(1)}, \ldots, X^{(l)}) \in (S_1 \union \cdots \union S_l)^m
$$
as the sequence that interleaves $X^{(1)}, \ldots, X^{(l)}$ according to the order given by $\tld{X}$: that is,
$
X_t = X^{\P{\tld{X}_t}}_{\sigma(t)}
$
where $\sigma(t) \ce \#\SETCO{t' \le t}{\tld{X}_{t'} = \tld{X}_t}$.
\end{definition}

Note that \cite{ChalermsookCS20} defines the \textit{decomposition} operation, which is the inverse operation of the composition. We will use \Cref{def:composition} throughout this paper.

\begin{definition}[$j_x$]
In the context of \Cref{def:composition}, for any key $x \in S_1 \union \cdots \union S_l$, let $j_x$ be the unique index such that $x \in S_{j_x}$.
\end{definition}

%% file: figs/geom-view.tex
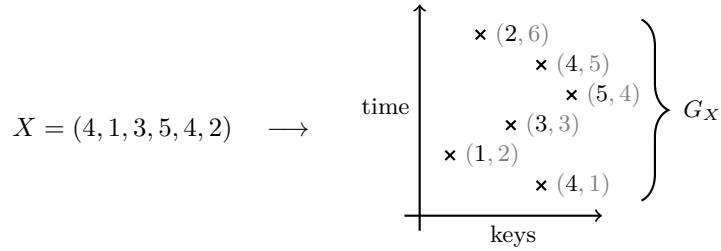
\begin{figure}[h]
\centering
\newcommand{\axes}[2]{
    \newcommand{\start}{-.5}
    \draw[->] (\start,0) -- (#1,0) node[pos=0.55, below] {keys};
    \draw[->] (0,\start) -- (0,#2) node[pos=0.55, left] {time};
}
\newcommand{\pointAndCoord}[2]{
\node[cross] (p1) at (#1,#2) {};
\node[gray, right=0mm of p1] {$({\color{black}{{#1}}},#2)$};
}
$
X = (4,1,3,5,4,2)
\quad \longrightarrow \quad
\vcenter{\hbox{
\begin{tikzpicture}[scale=0.5]
\axes{6}{7}
\pointAndCoord{4}{1}
\pointAndCoord{1}{2}
\pointAndCoord{3}{3}
\pointAndCoord{5}{4}
\pointAndCoord{4}{5}
\pointAndCoord{2}{6}
\draw [decorate,decoration={brace,amplitude=10pt},xshift=13mm,yshift=0pt]
(6,6.5) -- (6,0.5) node [midway,xshift=8mm] 
{$G_X$};
\end{tikzpicture}
}}
$
\caption{transforming $X$ into its geometric view $G_X$}\label{fig:geom-view}
\end{figure}

%% file: figs/mixvalue.tex
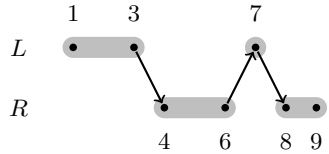
\begin{figure}[h]
\centering
\newcommand{\pointL}[1]{
    \node[dot] (p#1) at (#1,2) {};
    \node[above=1.5mm of p#1] {$#1$};
}
\newcommand{\pointR}[1]{
    \node[dot] (p#1) at (#1,0) {};
    \node[below=1.5mm of p#1] {$#1$};
}
\newcommand{\surround}[3]{
    \pgfmathsetmacro{\rad}{.35}
    \pgfmathsetmacro{\bottom}{#1-\rad}
    \pgfmathsetmacro{\top}{#1+\rad}
    \fill[rounded corners, lightgray] (#2-\rad,\bottom) rectangle (#3+\rad,\top) {};
}
\newcommand{\surroundL}[2]{\surround{2}{#1}{#2}}
\newcommand{\surroundR}[2]{\surround{0}{#1}{#2}}
\begin{tikzpicture}[scale=0.5]
\surroundL{1}{3}
\surroundL{7}{7}
\pointL{1}
\pointL{3}
\pointL{7}
\surroundR{4}{6}
\surroundR{8}{9}
\pointR{4}
\pointR{6}
\pointR{8}
\pointR{9}
\draw [->] (p3) -- (p4);
\draw [->] (p6) -- (p7);
\draw [->] (p7) -- (p8);
\node at (-0.8,2) {$L$};
\node at (-0.8,0) {$R$};
\end{tikzpicture}
\caption{a visualization of $\mixValue(\{1,3,7\},\{4,6,8,9\}) = 3$}\label{fig:mixvalue}
\end{figure}

%% file: figs/alt.tex
\begin{figure}[h]
\centering
\begin{tabu}{c}
\begin{tikzpicture}
\node[circle,draw] (u) at (3.125,4) {$u$};
\node[circle,draw] (v) at (1.75,3) {$v$};
\node[circle,draw] (w) at (4.5,3) {$w$};
\node[circle,draw] (x) at (2.5,2) {$x$};
\node[circle,draw] (k1) at (1,2) {$1$};
\node[circle,draw] (k2) at (2,1) {$2$};
\node[circle,draw] (k3) at (3,1) {$3$};
\node[circle,draw] (k4) at (4,2) {$4$};
\node[circle,draw] (k5) at (5,2) {$5$};
\newcommand{\arrowL}[2]{\draw[->] (#1) -- (#2) node[pos=.5, above left=-1mm and -1mm] {$\ttL{}$};}
\newcommand{\arrowR}[2]{\draw[->] (#1) -- (#2) node[pos=.5, above right=-1mm and -1mm] {$\ttR{}$};}
\newcommand{\children}[3]{\arrowL{#1}{#2}\arrowR{#1}{#3}}
\children{u}{v}{w}
\children{v}{k1}{x}
\children{w}{k4}{k5}
\children{x}{k2}{k3}
\end{tikzpicture}\\
reference tree $\cT$\\
\end{tabu}
\quad
\begin{tabu}{cccc}
Node & Link used by each access & Group by letter& $\#$\\
\midrule
$u$ & $\mathtt{R,L,L,R,R,L}$ & $\mathtt{[R],[L,L],[R,R],[L]}$ & 3\\%
$v$ & $\mathtt{L,R,R}$ & $\mathtt{[L],[R,R]}$ & 1\\
$w$ & $\mathtt{L,R,L}$ & $\mathtt{[L],[R],[L]}$ & 2\\
$x$ & $\mathtt{R,L}$ & $\mathtt{[R],[L]}$ & 1\\
\midrule
Total & & & 7\\
\end{tabu}

\caption{For access sequence $X = (4,1,3,5,4,2)$ and reference tree $\cT$, $\ALT_\cT(X) = 7$.}\label{fig:alt}
\end{figure}
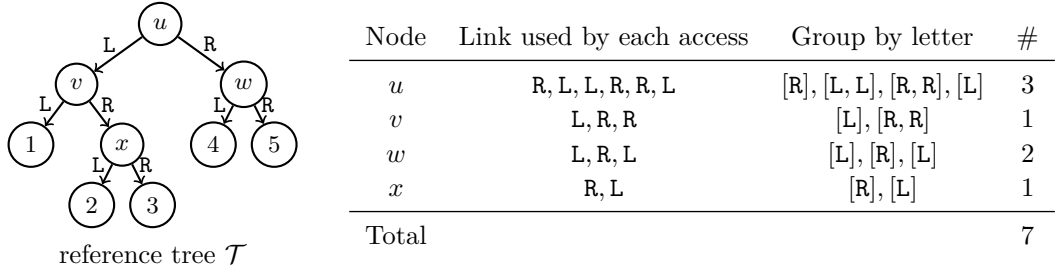

%% file: figs/fb.tex
\begin{figure}[h]
\centering
\newcommand{\diskColor}{black!40}
\newcommand{\shadowColor}{lightgray}
\newcommand{\bottom}{0}
\newcommand{\funnel}[3]{
    \fill[\diskColor] (#1,#2) circle (3mm);
    \node at (#1,#2) (C) {};
    \node[above=0mm of C] {#3};
}
\newcommand{\funnelL}[2]{
    \fill[lightgray] (#1,#2) rectangle (-0.5,\bottom);
    \funnel{#1}{#2}{$\ttL$}
}
\newcommand{\funnelR}[2]{
    \fill[lightgray] (#1,#2) rectangle (8.5,\bottom);
    \funnel{#1}{#2}{$\ttR$}
}

\newcommand{\access}[2]{\node[cross] (p#2) at (#1,#2) {};}
\begin{tabu}{c}
\begin{tikzpicture}[scale=0.5]
\funnelL{3}{3}
\funnelL{2}{7}
\funnelL{1}{8}
\funnelR{5}{4}
\funnelR{7}{6}
\fill[\shadowColor] (-0.5,\bottom) rectangle (8.5,1);
\access{4}{1}
\access{6}{2}
\access{3}{3}
\access{5}{4}
\access{1}{5}
\access{7}{6}
\access{2}{7}
\access{1}{8}
\access{4}{9}
\access{6}{10}
\access{3}{11}
\node[above=0mm of p9] {$p$};
\end{tikzpicture}\\
the funnel of $p$ has 5 points
\end{tabu}
\quad
\begin{tabu}{l}
By increasing $y$-coordinate: $\mathtt{L,R,R,L,L}$.\\
The side switches twice, so $f(P,p) = 2$.
\end{tabu}

\caption{Computing $f(P,p)$ for $p=(4,9)$ in the geometric view of $X=(4,6,3,5,1,7,2,1,4,6,3)$. Notice how the funnel points form a staircase-like front on either side of $p$.}\label{fig:fb}
\end{figure}

%% file: effect-of-composition.tex
\section{Effect of Composition on Wilber's bounds}
\label{sec:composition}
We prove \Cref{thm:direct sums} in this section. Namely, we show that Wilber's bounds act nicely under composition, allowing us to boost the separation between them in \Cref{sec:boosting}. We divide the proofs into the following two theorems. %

\begin{theorem}[subadditivity of $\Alt$]\label{thm:subadditivity-alt}
    Let $X \ce \tld{X}(X^{(1)}, \ldots, X^{(l)})$ be a composed sequence with $|X^{(1)}| = \cdots = |X^{(l)}|$.\footnote{We make this assumption so that the proof is simpler.} %
    Then
    $$
    \Alt(X) \le \Alt\P{\tld{X}} + \sum_j \Alt\P{X^{(j)}} + O(m).
    $$
\end{theorem}

We prove \Cref{thm:subadditivity-alt} in \Cref{sec:subaddib}.

\begin{theorem}[superadditivity of $\Funnel$]\label{thm:superadditivity-funnel}
    Let $X \ce \tld{X}(X^{(1)}, \ldots, X^{(l)})$ be a composed sequence. Then
    $$
    \Funnel(X) \ge \Funnel\P{\tld{X}} + \sum_j \Funnel\P{X^{(j)}} - O(m).
    $$
\end{theorem}

We prove \Cref{thm:superadditivity-funnel} in \Cref{sec:appendix:supperadditivity funnel}. 

\subsection{Subadditivity of the \IB{}} \label{sec:subaddib}
We prove \Cref{thm:subadditivity-alt} in this section.
\paragraph*{Proof plan}
We will show that for any binary search tree $\calT$ over $S_1 \union \cdots \union S_l$,
$$
\Alt_\calT(X) \le \Alt\P{\tld{X}} + \sum_j \Alt\P{X^{(j)}} + O(m).
$$
We will do this by
\begin{itemize}
    \item decomposing $\calT$ into the corresponding binary trees $\tld\calT$ over $[l]$ and $\calT_j$ over $S_j$ for all $j$;
    \item classifying preferred child alternations in $\calT$ into 4 types, which correspond to either
    \begin{itemize}
        \item alternations in $\tld\calT$,
        \item alternations in $\calT_j$ for some $j$,
        \item or to some other events that happen at most $O(m)$ times in aggregate.
    \end{itemize}
\end{itemize}
That is, we will show that
$$
\Alt_\calT(X) \le \Alt_{\tld\calT}\P{\tld{X}} + \sum_j \Alt_{\calT_j}\P{X^{(j)}} + O(m).
$$

\subsubsection{Decomposing the tree}
For a tree $\calT$, we write $x \prec_\calT b$ if $x$ is a descendent of $b$ in $\calT$, and we write $S \prec_\calT b$ if $x \prec_\calT b$ for all $x \in S$.

\begin{definition}\label{def:Tj}
    Let $\calT_j$ be the unique binary search tree over $S_j$ such that if $b,x \in S_j$ and $x \prec_\calT b$ then $x \prec_{\calT_j} b$.
\end{definition}

$\calT_j$ is constructed by running the following recursive algorithm, which builds a tree $\Tout$:
\begin{itemize}
\item Start at the root of $\calT$, and let $x$ be the current node.
\item If $x \in S_j$, then
    \begin{itemize}
	\item make $x$ the root of $\Tout$;
	\item form $x$'s left subtree in $\Tout$ by recursing on $x$'s left subtree in $\calT$;
	\item form $x$'s right subtree in $\Tout$ by recursing on $x$'s right subtree in $\calT$.
    \end{itemize}
\item If $x \not\in S_j$ then since $S_j$ is contiguous, at most one of $x$'s left and right subtrees in $\calT$ can contain elements from $S_j$.
    \begin{itemize}
	\item If there is one such subtree, form $\Tout$ by recursing on it.
	\item Otherwise let $\Tout$ be the empty tree.
    \end{itemize}
\end{itemize}

This algorithm clearly has the desired properties:
\begin{itemize}
\item Clearly, by construction, $\calT_j$ is a binary search tree and its set of keys is $S_j$.
\item If $b,x \in S_j$ and $x \prec_\calT b$, then $x \prec_{\calT_j} b$, because the only way to get to $x$ is to first pass through $b$, add it as a root of the current subtree, then recurse on $b$'s subtree that contains $x$, which eventually adds $x$ to $\calT_j$ as a descendent of $b$.
\item $\calT_j$ is unique since the only nontrivial choice the algorithm makes is to add $x$ as a root, but this is necessary since it must be an ancestor of all of the keys that later get added to this part of $\calT_j$.
\end{itemize}

\begin{definition}\label{def:tldT}
    Let $\tld\calT$ be the unique binary search tree over $[l]$ such that if $x \prec_\calT b$ and $S_{j_b}, S_{j_x} \prec_\calT b$ then $j_x \prec_{\tld\calT} j_b$.
\end{definition}
$\tld\calT$ is constructed by running the following recursive algorithm, which builds a tree $\Tout$:
\begin{itemize}
\item Start at the root of $\calT$, and let $x$ be the current node.
\item If $j_x$ hasn't already been seen earlier in the algorithm (which happens iff $x$ is the lowest common ancestor of all of $S_{j_x}$ in $\calT$), then
    \begin{itemize}
	\item make $j_x$ the root of $\Tout$;
	\item form $j_x$'s left subtree in $\Tout$ by recursing on $x$'s left subtree in $\calT$;
	\item form $j_x$'s right subtree in $\Tout$ by recursing on $x$'s right subtree in $\calT$.
    \end{itemize}
\item If $j_x$ has already been seen earlier in the algorithm, then some ancestor of $x$ was also in $S_{j_x}$, and $S_{j_x}$ is contiguous, so $S_{j_x}$ must contain either $x$'s entire left subtree, or $x$'s entire right subtree. That means that at most one of $x$'s subtrees can contain elements *not* in $S_{j_x}$.
    \begin{itemize}
	\item If there is one such subtree, form $\Tout$ by recursing on it.
	\item Otherwise, let $\Tout$ be the empty tree.
    \end{itemize}
\end{itemize}

This algorithm clearly has the desired properties:
\begin{itemize}
\item Clearly, by construction, $\tld\calT$ is a binary search tree and its set of keys is $[l]$.
\item If $x \prec_\calT b$ and $S_{j_b}, S_{j_x} \prec_\calT b$ then:
    \begin{itemize}
	\item We can assume $j_x \ne j_b$ and thus $x \ne b$, otherwise the claim is trivially true.
	\item Since $S_{j_b} \prec_\calT b$, $b$ must be the lowest common ancestor of all of $S_{j_b}$ in $\calT$, so $j_b$ gets added to $\tld\calT$ when the algorithm is looking at node $b$.
	\item Also, since $S_{j_x} \prec_\calT b$ and $j_x \ne j_b$, that means that all the elements from $S_{j_x}$ are descendents of $b$, so $j_x$ will be added to $\tld\calT$ in one of the recursive branches launched when looking at $b$.
	\item Therefore $j_x$ will be a descendent of $j_b$ in $\tld\calT$.
    \end{itemize}
\item $\tld\calT$ is unique since the only nontrivial choice the algorithm makes is to add $j_x$ as root when it first sees an element from $S_{j_x}$, but this is necessary since for any $j$ which will eventually be added to this part of the tree, $S_j$ must have been completely contained in $x$'s subtree, and therefore $j_x$ must be an ancestor of $j$.
\end{itemize}

\subsubsection{Stating the classification}
Consider some left-to-right\footnote{The case where the alternation occurs from right to left is analogous.} preferred child alternation that $X$ produces in $\calT$.
That is, take some value $b \in S_1 \union \cdots \union S_l$ and some times $t_x < t_y \in [m]$ such that
\begin{itemize}
    \item $x \ce X_{t_x}$ is in the left subtree of $b$ in $\calT$,
    \item $y \ce X_{t_y}$ is in the right subtree of $b$ in $\calT$,
    \item and none of the accesses $X_{t_x+1}, \ldots, X_{t_y-1}$ made in the interim were to values that are strict descendents of $b$.
\end{itemize}
Let $a$ be the lowest ancestor of $b$ such that $a<b$ and $c$ be the lowest ancestor of $b$ such that $b<c$.\footnote{If either $a$ or $c$ doesn't exist, let $a \ce -\infty$ or $c \ce +\infty$ by convention.} This means that the left and right subtrees of $b$ correspond to the keys in intervals $(a,b)$ and $(b,c)$. We have $x,y \prec_\calT b$, $x \in (a,b)$, and $y \in (b,c)$.

\begin{claim}\label{clm:classification}
One of the following must hold (from most ``local'' to most ``global''):
\begin{enumerate}
\item All of $b,x,y$ are in the same range $S_{j_b}$, so $x$ and $y$ are in the left and right subtrees of $b$ in $\calT_{j_b}$, and this corresponds to an alternation in $\calT_{j_b}$.
\item $b$ is either the highest ancestor of $x$ such that $b \in S_{j_x}$ and $x<b$, or the highest ancestor of $y$ such that $b \in S_{j_y}$ and $b < y$. %
\item $S_{j_b} \prec_\calT b$, and either $j_x$ is the closest (in key value) ancestor of $j_b$ in $\tld\calT$ such that $j_x<j_b$, or $j_y$ is the closest (in key value) ancestor of $j_b$ in $\tld\calT$ such that $j_b < j_y$.
\item All of $b,x,y$ are in different ranges (i.e. $j_x < j_b < j_y$), $j_x$ is in $j_b$'s left subtree in $\tld\calT$, and $j_y$ is in $j_b$'s right subtree in $\tld\calT$, so this corresponds to an alternation in $\tld\calT$.
\end{enumerate}
\end{claim}

\input{figs/classification}

\subsubsection{Proving the classification}
\begin{proof}[Proof of \Cref{clm:classification}]
Let us prove that every alternation is of one of these four types.
\begin{itemize}
\item First, suppose that $j_x = j_b = j_y$. Then by construction of $\calT_{j_b}$, $x$ and $y$ are still descendents of $b$ in $\calT_{j_b}$, and since $\calT_{j_b}$ is a binary search tree, $x$ must be in $b$'s left subtree and $y$ must be in $b$'s right subtree. So this is type 1.
\item Now suppose that exactly one of $j_x = j_b$ and $j_b = j_y$ holds. By symmetry, suppose that it is the former, and thus $j_x = j_b < j_y$. Then we trivially have $b \in S_{j_x}$. And on the other hand, consider any ancestor $b'$ of $x$ which is higher than $b$ and satisfies $x<b'$. Then $b'$ would have to satisfy $y < b'$ as well, and in particular $j_y \le j_{b'}$, so it could not lie in $S_{j_x}$. Therefore $b$ is the highest ancestor of $x$ which lies in $S_{j_x}$ and satisfies $x<b$. So this is type 2.
\item From now on we can assume that $j_x < j_b < j_y$, which means in particular that $j_a < j_b < j_c$, that $S_{j_b}$ is contained entirely in $b$'s subtree in $\calT$, and therefore $b$ is the highest member of $S_{j_b}$ in $\calT$. Now, the lowest common ancestor of $S_{j_a}$ (resp. $S_{j_c}$) must be an ancestor of $a$ (resp. $c$) and therefore an ancestor of $b$, so by the properties of $\tld\calT$, $j_a$ (resp. $j_c$) is an ancestor of $j_b$ in $\tld\calT$. Furthermore, any ancestor of $j_b$ in $\tld\calT$ must be of the form $j_{z}$ for some ancestor $z$ of $b$ in $\calT$, so since $a$ (resp. $c$) is the closest (in key value) ancestors of $b$ on its left (resp. right) side in $\calT$, $j_a$ (resp. $j_c$) must be the closest ancestor of $j_b$ on its left (resp. right) side in $\tld\calT$.
    \begin{itemize}
	\item Suppose that at least one of $j_a = j_x$ or $j_y = j_c$ holds. By symmetry suppose that it is the former. Then just by virtue of the fact that $j_x = j_a$, $j_x$ is the closest ancestor of $j_b$ on its left side in $\tld\calT$. So this is type 3.
	\item Otherwise, we have $j_a < j_x < j_b < j_y < j_c$. This implies that $S_{j_x}$ lies entirely within $b$'s left subtree, and $S_{j_y}$ lies entirely within $b$'s right subtree, thus $j_x$ and $j_y$ are descendents of $j_b$ in $\tld\calT$. So this is type 4.\qedhere
    \end{itemize}
\end{itemize}
\end{proof}

\subsubsection{Using the classification to prove \Cref{thm:subadditivity-alt}}
\begin{proof}[Proof of \Cref{thm:subadditivity-alt}]
Let $\calT$ be any binary search tree over $S_1 \union \cdots \union S_l$, and let $\calT_j$ and $\tld\calT$ be the corresponding trees defined in \Cref{def:Tj} and \Cref{def:tldT}. We will show that 
$$
\Alt_\calT(X) \le \Alt_{\tld\calT}\P{\tld{X}} + \sum_j \Alt_{\calT_j}\P{X^{(j)}} + O(m).
$$
Let us use \Cref{clm:classification}: we charge type 1 alternations to $\Alt_{\calT_j}\P{X^{(j)}}$, type 4 alternations to $\Alt_{\tld\calT}\P{\tld{X}}$, and we show below that there are only $O(m)$ alternations of types 2 and 3.

For type 2, this is because we can charge it uniquely to the access made to $x$ or $y$ (formally, we charge it to $t_x$ or $t_y$).
\begin{itemize}
\item Let us take the first subcase: $b$ is the highest ancestor of $x$ such that $b \in S_{j_x}$ and $x<b$. $x$ can only have one highest ancestor with a given property, so it has only one highest ``ancestor $b$ such that $b \in S_{j_x}$ and $x<b$''. So this can apply to at most one of the alternations that occurred when accessing $x$, and thus we can charge it to $t_x$.
\item Let us take the second subcase: $b$ is the highest ancestor of $y$ such that $b \in S_{j_y}$ and $b < y$. Again, $y$ can only have one highest ``ancestor $b$ such that $b \in S_{j_y}$ and $b < y$''. The access to $x$ is the first time that the preferred child switches back from $b$'s left child to $b$'s right child after accessing $y$. So this event is unique from the perspective of this particular access to $y$, and thus we can charge it to $t_y$.
\end{itemize}

Finally, we will bound the total number of occurrences of type 3 alternations by charging them to $j_b$, not uniquely but in a $\f{l}{m}$-to-$1$ manner. Let us take the case where $j_x$ is the closest ancestor of $j_b$ in $\tld{\calT}$ such that $j_x<j_b$ (the other case is analogous). Clearly, $j_b$ determines $j_x$ uniquely. And since $b$'s subtree contains $S_{j_b}$ in its entirety, $j_b$ determines $b$ uniquely too. So once you know $j_b$, the only uncertainty left about this alternation is \emph{which} access within $X^{(j_x)}$ caused it. So the total number of alternations of this type is bounded by
$$
\ub{l}_\text{which $j_b$?}\ \ub{\max_{j_x} \ABS{X^{(j_x)}}}_\text{which access within $X^{(j_x)}$?} = l \f{m}{l} = m,
$$
where the first equality uses our assumption that $|X^{(1)}| = \cdots = |X^{(l)}|$.

Overall, we have shown that
$$
\Alt_\calT(X) \le \ub{\Alt_{\tld\calT}\P{\tld{X}}}_\text{type 4} + \sum_j \ub{\Alt_{\calT_j}\P{X^{(j)}}}_\text{type 1} + \ub{O(m)}_\text{type 2 (charge to $t_x$ or $t_y$)} + \ub{O(m)}_\text{type 3 (charge to $j_b$)},
$$
so we can now take the maximum over $\calT$ to conclude
$$
\begin{aligned}
    \Alt(X)
    &\ce \max_{\calT} \Alt_{\calT}(X)\\
    &\le \max_{\calT} \p*{\Alt_{\tld\calT}\P{\tld{X}} + \sum_j \Alt_{\calT_j}\P{X^{(j)}} + O(m)}\\
    &\le \max_{\tld\calT} \Alt_{\tld\calT}\P{\tld{X}} + \sum_j \max_{\calT_j}\Alt_{\calT_j}\P{X^{(j)}} + O(m)\\
    &= \Alt\P{\tld{X}} + \sum_j \Alt\P{X^{(j)}} + O(m).\qedhere
\end{aligned}
$$
\end{proof}

%% file: figs/classification.tex
\begin{figure}[h]
\newcommand{\drawTree}{
\node (a) at (-0.5,5) {$a$};
\node (d) at (2.5,4) {$\bullet$};
\node (c) at (2,3) {$c$};
\node (b) at (1,2) {$b$};
\node (x) at (0.5,1) {$x$};
\node (y) at (1.5,1) {$y$};
\draw (a) -- (d);
\draw [dashed] (d) -- (c);
\draw (c) -- (b);
\draw [dashed] (b) -- (x);
\draw [dashed] (b) -- (y);
}

\centering
\begin{tabu}{cccc}
type 1 & type 2 & type 3 & type 4\\
\begin{tikzpicture}
\drawTree
\draw [decorate,decoration={brace,amplitude=5pt},yshift=-3mm]
(1.7,1) -- (0.3,1) node [midway,yshift=-5mm] 
{$S_{j_b}$};
\end{tikzpicture}
&
\begin{tikzpicture}
\drawTree
\draw [decorate,decoration={brace,amplitude=5pt},yshift=-3mm]
(1.2,1) -- (0.3,1) node [midway,yshift=-5mm] 
{$S_{j_b}$};
\end{tikzpicture}
&
\begin{tikzpicture}
\drawTree
\draw [decorate,decoration={brace,amplitude=5pt},yshift=-3mm]
(0.7,1) -- (-0.7,1) node [midway,yshift=-5mm] 
{$S_{j_a}$};
\draw [decorate,decoration={brace,amplitude=5pt},yshift=-3mm]
(1.2,1) -- (0.8,1) node [midway,yshift=-5mm] 
{$S_{j_b}$};
\end{tikzpicture}
&
\begin{tikzpicture}
\drawTree
\draw [decorate,decoration={brace,amplitude=5pt},yshift=-3mm]
(0.2,1) -- (-0.7,1) node [midway,yshift=-5mm] 
{$S_{j_a}$};
\draw [decorate,decoration={brace,amplitude=5pt},yshift=-3mm]
(1.2,1) -- (0.8,1) node [midway,yshift=-5mm] 
{$S_{j_b}$};
\draw [decorate,decoration={brace,amplitude=5pt},yshift=-3mm]
(2.7,1) -- (1.8,1) node [midway,yshift=-5mm] 
{$S_{j_c}$};
\end{tikzpicture}
\\
charge to $\calT_{j_b}$ & charge to $t_x$ & charge to $j_b$ & charge to $\tld\calT$\\
\end{tabu}

\caption{A preferred child alternation in $\calT$: $b$'s preferred child changes from the left side (due to an access to $x$) to the right side (due to an access to $y$). There are four qualitatively different ways in which the 
alternation can happen depending on which ranges $S_1, \ldots, S_l$ the keys $a,b,c,x,y$ belong.}\label{fig:classification}
\end{figure}
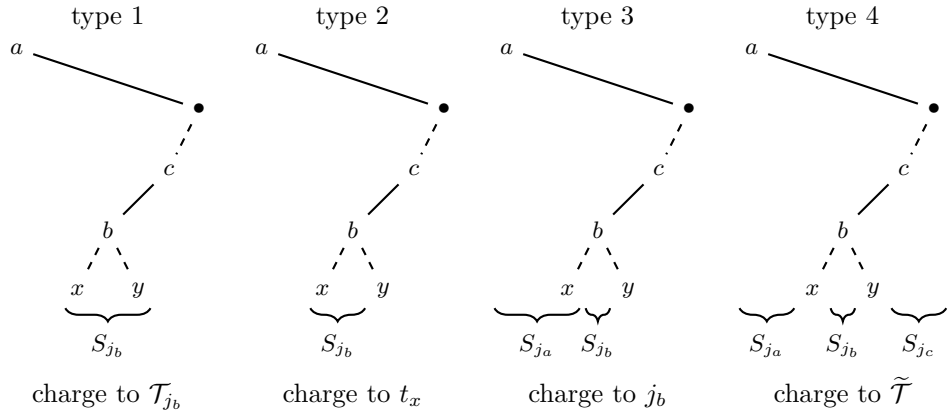

%% file: boosting-separation.tex
\section{Boosting the separation between Wilber's bounds}
\label{sec:boosting}
We prove \Cref{thm:main_boosting} in this section. We now use the composition properties of $\Alt$ and $\Funnel$ we proved in \Cref{sec:composition} to show that Tango tree makes an optimal trade-off between fixed costs and variable costs that depend on the alternation bound.

\subsection{What boosting can we get?}
Lecomte and Weinstein \cite{lecomte2019settling} show an $\Omega(\log \log n)$ separation between $\Alt$ and $\Funnel$.

\begin{theorem}[Theorem~2 in~\cite{lecomte2019settling}] \label{thm:separation n}
For any $n$ of the form $2^{2^k}$, there is a sequence $Y_n \in [n]^m$ where $m \le \poly(n)$, each element appears $O(m/n)$ times, and
$$
\al{
\ol\Alt(Y_n) &\le O(1) \\
\ol\Funnel(Y_n) &\ge \Omega(\log \log n).
}
$$
\end{theorem}

We can use the tight composition results from \Cref{sec:composition} to show the following boosted separation. We emphasize that the approximate subadditivity of the Alternation bound~\cite{ChalermsookCS20} is insufficient to boost the separation.  

\begin{theorem}[Hardness Amplification, Restatement of \Cref{thm:main_boosting}]\label{thm:boosted-separation}
There is a constant $K>0$ such that for any $n$ of the form $2^{2^k}$ and any power-of-two $R \le \f{\log n}{K}$, there is a sequence $Y_n^{\circ R} \in [n]^{m'}$ with $m' \le \poly(n)$ such that
$$
\al{
\ol\Alt\P{Y_n^{\circ R}} &\le O(R)\\
\ol\Funnel\P{Y_n^{\circ R}} &\ge \Omega\P{R\log \P{\f{\log n}{R}}}.
}
$$
\end{theorem}
\begin{proof}
Let $Y_n$ be the sequence stated in \Cref{thm:separation n}. First, pad $Y_n$ so that each key appears \emph{exactly} $\abs{Y_n}/n$ times, by adding each key one by one the appropriate number of times, in ascending order. It is easy to see that this maintains the bounds
$$
\al{
\ol\Alt(Y_n) &\le O(1)\\
\ol\Funnel(Y_n) &\ge \Omega(\log \log n).
}
$$
Now, let $C_O, C_\Omega>0$ be constants such that
$$
\al{
\ol\Alt(Y_n) &\le C_O\\
\ol\Funnel(Y_n) &\ge C_\Omega\log \log n
}
$$
(we will allow ourselves to make $C_O$ even larger later on).

Let $Y_n^{\circ 1} \ce Y_n$, and for all power-of-two $R \ge 1$, let
$$
Y_n^{\circ 2R} \ce \p*{Y_{\sqrt{n}}^{\circ R}}^{\otimes \sqrt{n}}\P{Y_{\sqrt{n}}^{\circ R}, \ldots, Y_{\sqrt{n}}^{\circ R}},
$$
where \quotes{$X^{\otimes \sqrt{n}}$} means \quotes{$X$ repeated $\sqrt{n}$ times}, and with an abuse of notation we assume that the $\sqrt{n}$ sequences $Y_{\sqrt{n}}^{\circ R}, \ldots, Y_{\sqrt{n}}^{\circ R}$ that are being composed each contains a distinct range of keys. 
We can check that
$$
\abs*{Y_n^{\circ R}} = n^{1-1/R} |Y_{n^{1/R}}| \le n^{1-1/R}\ \poly\P{n^{1/R}} \le \poly(n)
$$
as desired. 
We will show by induction that
$$
\al{
\ol\Alt\P{Y_n^{\circ R}} &\le C_O(2R-1)\\
\ol\Funnel\P{Y_n^{\circ R}} &\ge C_\Omega\f{R+1}2 \log\P{ \f{\log n}R}.
}
$$
\paragraph*{Base case: $R=1$.} We verify that indeed
$$
\al{
\ol\Alt(Y_n) &\le C_O\\
&= C_O(2\cdot 1 - 1)\\
&= C_O(2R-1)
}
$$
and
$$
\al{
\ol\Funnel(Y_n)
&\ge C_\Omega\log \log n\\
&= C_\Omega \f{1+1}{2} \log\P{ \f{\log n}1}\\
&= C_\Omega \f{R+1}{2} \log\P{ \f{\log n}R}.
}
$$
\paragraph*{Inductive case: $R \to 2R$.} Suppose this is true for some $R \ge 1$, for all $n$. Then for $\Alt$, by \Cref{thm:subadditivity-alt} we have
$$
\al{
\ol\Alt\P{Y_n^{\circ 2R}} &\le \ol\Alt\P{Y_{\sqrt{n}}^{\circ R}} + \f{\sqrt{n} \cdot \ol\Alt\P{Y_{\sqrt{n}}^{\circ R}}}{\sqrt{n}} + O(1)\\
&\leq  2C_O\P{2R-1} + O(1)\\
&= C_O\p*{4R-1} - C_O + O(1)\\
&\le C_O\p*{4R-1},
}
$$
where the last step holds as long as $C_O$ is large enough.

For $\Funnel$, by \Cref{thm:superadditivity-funnel} we have
$$
\al{
\ol\Funnel\P{Y_{\sqrt{n}}^{\circ R}} &\ge \ol\Funnel\P{Y_{\sqrt{n}}^{\circ R}} + \f{\sqrt{n} \cdot \ol\Funnel\P{Y_{\sqrt{n}}^{\circ R}}}{\sqrt{n}} - O(1)\\
&\geq  C_\Omega(R+1) \log\P{\f{\log \sqrt{n}}R} - O(1)\\
&= C_\Omega\f{2R+1}{2} \log\P{\f{\log n}{2R}} + \f{C_\Omega}{2}\log\P{\f{\log n}{2R}} - O(1)\\
&\ge C_\Omega\f{2R+1}{2} \log\P{\f{\log n}{2R}} + \f{C_\Omega}{2}\log\P{\f{K}{2}} - O(1)\\
&\ge C_\Omega\f{2R+1}{2} \log\P{\f{\log n}{2R}},
}
$$
where the penultimate step holds because $R \le \frac{\log n}{K}$, and the last step holds as long as $K$ is large enough.
\end{proof}

\subsection{Optimality of Tango Trees} \label{sec:opt tango}
We are  ready to prove \Cref{thm:lower bound alt charging}.  
\begin{theorem}[Restatement of \Cref{thm:lower bound alt charging}]
    Let $\alpha,\beta: \N \to \R_{\ge 1}$ be any functions and let $A$ be some BST algorithm. Denote the amortized cost of $A$ on an access sequence $X$ as $\ol A(X)$. Suppose that for all access sequences $X$ over $n$ keys, we have
    $$
    \ol A(X) \le \alpha(n)\ol\Alt(X) + \beta(n).
    $$
    Then $\alpha(n) \ge \Omega\P{\log \P{\f{\log n}{\beta(n)}}}$ for infinitely many values of $n$.
\end{theorem}
\begin{proof}
In fact we will show that the result holds under the weaker assumption that the theorem holds for the \emph{optimal} amortized cost:
$$
\ol \OPT(X) \le \alpha(n)\ol\Alt(X) + \beta(n).
$$
The above inequality must in particular hold for the access sequence $Y_n^{\circ R}$ from \Cref{thm:main_boosting}, where we let $R$ be the largest power of two such that $R\le\beta(n)$. This gives us
$$
\left\{
\al{
\ol \OPT\P{Y_n^{\circ R}} &\le \alpha(n)\ol\Alt(X) + \beta(n) \le O\P{R(\alpha(n) + 1)}\\
\ol \OPT\P{Y_n^{\circ R}} &\ge \Omega\P{\ol\Funnel\P{Y_n^{\circ R}}} \ge \Omega\P{R\log \P{\f{\log n}{R}}}.
}
\right.
$$
Combining these inequalities, we obtain
$$
R(\alpha(n) + 1) \ge \Omega\P{R\log \P{\f{\log n}{R}}} \implies \alpha(n) \ge \Omega\P{\log \P{\f{\log n}{R}}} \ge \Omega\P{\log \P{\f{\log n}{\beta(n)}}}
$$
(where the implication uses the fact that $\alpha(n)\ge 1$).
\end{proof}


\paragraph*{Optimal Trade-off } The matching trade-offs for $\alpha(n)$ and $\beta(n)$ follows from Tango trees~\cite{DHIP07} as discussed in the introduction. We provide an alternate proof by designing an algorithm parameterized by $k$ in \Cref{lem:bootstrapping tango} for completeness. That is, we describe a family of simple algorithms parameterized by $k \leq \log n$ that  gives a matching upper bound of
$$
O\P{ \P{\ALT(X) + m \cdot \frac{\log n}{k}} \cdot (\log k + 1)}.
$$

In this algorithm, the fixed cost is $\beta(n) = \Theta\P{\f{\log n \log k}{k}}$ and multiplicative factor $\alpha(n) = \Theta(\log k)$. We just proved that any algorithm with fixed cost $\beta(n)$ must have
$$
\alpha(n) \ge \Omega\P{\log \P{\f{\log n}{\beta(n)}}} = \Omega\P{\log \P{\f{k}{\log k}}} = \Omega(\log k),
$$
so our trade-off is optimal for any sufficiently large $k \le \log n$.

\emph{Note:} The trade-off above suggests that there is some leeway in the other direction, i.e. we could have an algorithm with the same multiplicative factor of $O(\log k)$ but a reduced fixed cost of $\f{\log n}{(\log k)^C}$ for any $C>1$ we choose. And this is indeed the case: we simply need to change variables
$$
k' \ce k^{1/C} \implies O\P{(\log k)\ol\Alt(X) + \f{\log n \log k}{k}} \le O\P{\p*{\log k'}\ol\Alt(X) + \f{\log n}{\p*{k'}^{C - o(1)}}}.
$$

\begin{theorem}[Restatement and full version of \Cref{thm:lower bound alt as function of opt}] \label{thm:full version of thm lower bound as f of opt}
Let $\alpha: \N\times\R_{\ge 1} \to \R_{\ge 1}$ be any function and let $A$ be some BST algorithm. Denote the amortized cost of $A$ on an access sequence $X$ as $\ol A(X)$. Suppose that for all access sequences $X$ over $n$ keys, we have
$$
\ol A(X) \le \alpha\P{n,\ol\OPT(X)} \cdot \P{\ol\Alt(X)+1}.
$$
Then $\alpha(n,s) \ge \Omega\P{\log \P{\f{\log n}{s}}}$ for infinitely many values of $n$, and (for each such $n$) for a set of values for $s$ that ranges from $O(1)$ to $\Omega(\log n)$ and are spaced by at most a constant factor.
\end{theorem}

\begin{proof}
Again, we will show that the result holds under the weaker assumption that the theorem holds for the \emph{optimal} amortized cost:
$$
\ol \OPT(X) \le \alpha\P{n,\ol\OPT(X)} \cdot \P{\ol\Alt(X)+1}.
$$
The above inequality must in particular hold for the access sequence $Y_n^{\circ R}$ from \Cref{thm:main_boosting}, for any power-of-two $R$. This gives us
$$
\left\{
\al{
\ol \OPT\P{Y_n^{\circ R}} &\le \alpha\P{n,\ol\OPT(Y_n^{\circ R})} \cdot \P{\ol\Alt(Y_n^{\circ R})+1} \le O\P{R\alpha\P{n,\ol\OPT(Y_n^{\circ R})}}\\
\ol \OPT\P{Y_n^{\circ R}} &\ge \Omega\P{\ol\Funnel\P{Y_n^{\circ R}}} \ge \Omega\P{R\log \P{\f{\log n}{R}}}.
}
\right.
$$
Combining these inequalities, we obtain
$$
\alpha\P{n,\ol\OPT(Y_n^{\circ R})} \ge \Omega\P{\log \P{\f{\log n}{R}}} \ge \Omega\P{\log \P{\f{\log n}{\ol\ALT(Y_n^{\circ R})}}} \ge \Omega\P{\log \P{\f{\log n}{\ol\OPT(Y_n^{\circ R})}}}.
$$
It is easy to verify that $\ol\OPT(Y_n^{\circ R})$ can range from $O(1)$ to $\Omega(\log n)$ and that the available values spaced by at most a constant factor. 
\end{proof}

\section{Superadditivity of the Funnel Bound} \label{sec:appendix:supperadditivity funnel}

We prove \Cref{thm:superadditivity-funnel} in this section.

\begin{proof}[Proof of~\Cref{thm:superadditivity-funnel}]
The proof is mostly by picture (see \Cref{fig:internal-external-funnel}).
Take some point $p=(x,t)$ in $G_X$, the geometric view of $X$. Its funnel will include:
\begin{itemize}
\item Its \emph{external funnel}: the points corresponding to its funnel from $\tld{X}$.\footnote{Or, more precisely, corresponding to the funnel of $(j_x, t)$ in the geometric view $G_{\tld{X}}$.} This is because e.g. if $j \in [l]$ is the most recent access that's close to $j_x$ on its left side, then the corresponding access in the composed sequence will definitely be the most recent access that's close to $x$ on its left side, too.
\item Its \emph{internal funnel}: the accesses corresponding to its funnel points from $X^{(j_x)}$.\footnote{Or, more precisely, corresponding to the funnel of $(x, \sigma(t))$ in the geometric view $G_{X^{(j_x)}}$, where $\sigma$ is defined as in \Cref{def:composition}.} This is because e.g. if $y \in S_{j_x}$ is the most recent access that's close to $x$ on its left side, then the other sequences $X^{(j)}$ for $j \ne j_x$ are not going to interfere with that.
\item And maybe some extra funnel points.
\end{itemize}

\input{figs/internal-external-funnel}

In addition, all points of the external funnel are later in time than all points of the internal funnel. Suppose by contradiction there is an internal funnel point $(x_I, t_I)$ that comes later than an external funnel point $(x_E, t_E)$, i.e., $t_I > t_E$. Also assume w.l.o.g.~that $j_{x_E} < j_x$. This means in $G_{\widetilde{X}}$ we must have $j_x = j_{x_I}$, then it means the point $(j_{x_I}, t_I)$ is in the rectangle formed by $(j_x, t)$ and $(j_{x_E}, t_E)$ since $t > t_I > t_E$ and $j_x = j_{x_I} > j_{x_E}$, and this contradicts with the definition that $(j_{x_E}, t_E)$ is a funnel point in $G_{\widetilde{X}}$.

So every side switch between consecutive points of the external funnel gives at least one side switch in the combined funnel, and so does every side switch between consecutive points of the internal funnel, and there is no double counting.

So the number of times the funnel points switch sides is at least the number of times the funnel points switch sides in $\tld{X}$ and $X^{(j_x)}$ combined. Thus
$$
\Funnel(X) \ge \Funnel\P{\tld{X}} + \sum_j \Funnel\P{X^{(j)}},
$$
and we don't even suffer an $O(m)$ loss.\footnote{A $O(m)$ loss could have arisen depending on the precise definition of the funnel bound, e.g. if as in \cite{lecomte2019settling} we had counted the number of \emph{consecutive streaks} of one side instead of counting the number of times the funnel \emph{flips} sides.}
\end{proof}

%% file: figs/internal-external-funnel.tex
\begin{figure}[h]
\centering
\newcommand{\diskColor}{black!40}
\newcommand{\inColor}{blue!50}
\newcommand{\exColor}{red!50}
\newcommand{\shadowColor}{lightgray}
\newcommand{\funnel}[4]{
    \fill[#4] (#1,#2) circle (3mm);
    \node at (#1,#2) (C) {};
    \node[above=0mm of C] {#3};
}
\newcommand{\funnelL}[3]{
    \fill[\shadowColor] (#1,#2) rectangle (\leftBorder,\bottom);
    \funnel{#1}{#2}{$\ttL$}{#3}
}
\newcommand{\funnelR}[3]{
    \fill[\shadowColor] (#1,#2) rectangle (\rightBorder,\bottom);
    \funnel{#1}{#2}{$\ttR$}{#3}
}
\newcommand{\access}[2]{\node[cross] (p#2) at (#1,#2) {};}
\newcommand{\separator}[1]{\draw [dashed] (#1,\topBorder) -- (#1,\bottom);}

\begin{tabu}{ccc}
\begin{tikzpicture}[scale=0.5]
\newcommand{\bottom}{-1}
\newcommand{\topBorder}{12.5}
\newcommand{\leftBorder}{-3.5}
\newcommand{\rightBorder}{8.5}
\funnelR{5}{2}{\inColor}
\funnelL{3}{3}{\inColor}
\funnelR{6}{5}{\diskColor}
\funnelL{2}{7}{\exColor}
\funnelL{-1}{8}{\diskColor}
\funnelR{7}{10}{\exColor}
\funnelL{-2}{11}{\exColor}
\fill[\shadowColor] (\leftBorder,\bottom) rectangle (\rightBorder,0);
\separator{-0.5}
\separator{2.5}
\separator{5.5}
\access{4}{0}
\access{0}{1}
\access{5}{2}
\access{3}{3}
\access{1}{4}
\access{6}{5}
\access{8}{6}
\access{2}{7}
\access{-1}{8}
\access{-3}{9}
\access{7}{10}
\access{-2}{11}
\access{4}{12}
\node[above=0mm of p12] {$(x,t)$};
\draw [decorate,decoration={brace,amplitude=5pt},yshift=-3mm]
(-0.7,\bottom) -- (-3.3,\bottom) node [midway,yshift=-5mm] 
{$S_{j_x-2}$};
\draw [decorate,decoration={brace,amplitude=5pt},yshift=-3mm]
(2.3,\bottom) -- (-0.3,\bottom) node [midway,yshift=-5mm] 
{$S_{j_x-1}$};
\draw [decorate,decoration={brace,amplitude=5pt},yshift=-3mm]
(5.3,\bottom) -- (2.7,\bottom) node [midway,yshift=-5mm] 
{$S_{j_x}$};
\draw [decorate,decoration={brace,amplitude=5pt},yshift=-3mm]
(8.3,\bottom) -- (5.7,\bottom) node [midway,yshift=-5mm] 
{$S_{j_x+1}$};
\draw [->, \inColor] (p2) -- (p3);
\draw [->, \diskColor] (p3) -- (p5);
\draw [->, \diskColor] (p5) -- (p7);
\draw [->, \exColor] (p8) -- (p10);
\draw [->, \exColor] (p10) -- (p11);
\end{tikzpicture}

&

\begin{tikzpicture}[scale=0.5]
\newcommand{\bottom}{-1}
\newcommand{\topBorder}{11.5}
\newcommand{\leftBorder}{-0.5}
\newcommand{\rightBorder}{3.5}
\funnelL{1}{7}{\exColor}
\funnelR{3}{10}{\exColor}
\funnelL{0}{11}{\exColor}
\fill[\shadowColor] (\leftBorder,\bottom) rectangle (\rightBorder,3);
\access{2}{0}
\access{1}{1}
\access{2}{2}
\access{2}{3}
\access{1}{4}
\access{3}{5}
\access{3}{6}
\access{1}{7}
\access{0}{8}
\access{0}{9}
\access{3}{10}
\access{0}{11}
\access{2}{12}
\node[above=0mm of p12] {$(j_x,t)$};
\draw [decorate,decoration={brace,amplitude=5pt},yshift=-3mm]
(3.3,\bottom) -- (-0.3,\bottom) node [midway,yshift=-5mm] 
{$[l]$};
\draw [->, \exColor] (p7) -- (p10);
\draw [->, \exColor] (p10) -- (p11);
\end{tikzpicture}

&

\begin{tikzpicture}[scale=0.5]
\newcommand{\bottom}{-1}
\newcommand{\topBorder}{11.5}
\newcommand{\leftBorder}{2.5}
\newcommand{\rightBorder}{5.5}
\funnelR{5}{1}{\inColor}
\funnelL{3}{2}{\inColor}
\fill[\shadowColor] (\leftBorder,\bottom) rectangle (\rightBorder,0);
\access{4}{0}
\access{5}{1}
\access{3}{2}
\access{4}{3}
\node[above=0mm of p3] {$(x,\sigma(t))$};
\draw [decorate,decoration={brace,amplitude=5pt},yshift=-3mm]
(5.3,\bottom) -- (2.7,\bottom) node [midway,yshift=-5mm] 
{$S_{j_x}$};
\draw [->, \inColor] (p1) -- (p2);
\end{tikzpicture}
\\
$G_X$
&
$G_{\tld{X}}$
&
$G_{X^{(j_x)}}$
\\
\end{tabu}

\caption{The funnel of an access in the combined sequence $X$ (left) includes the \emph{external funnel} (in red) of the corresponding access in $\tld{X}$ (center) and the \emph{internal funnel} (in blue) of the corresponding accesss in $X^{(j_x)}$ (right), as well as some extra funnel points corresponding to neither (in gray). Therefore, the contribution of $(x,t)$ to $\Funnel(X)$ (arrows on the left), is lower bounded by the contribution of $(j_x,t)$ to $\Funnel\P{\tld{X}}$ (red arrows) and the contribution of $(x,\sigma(t))$ to the $\Funnel\P{X^{(j_x)}}$ (blue arrows).}\label{fig:internal-external-funnel}
\end{figure}
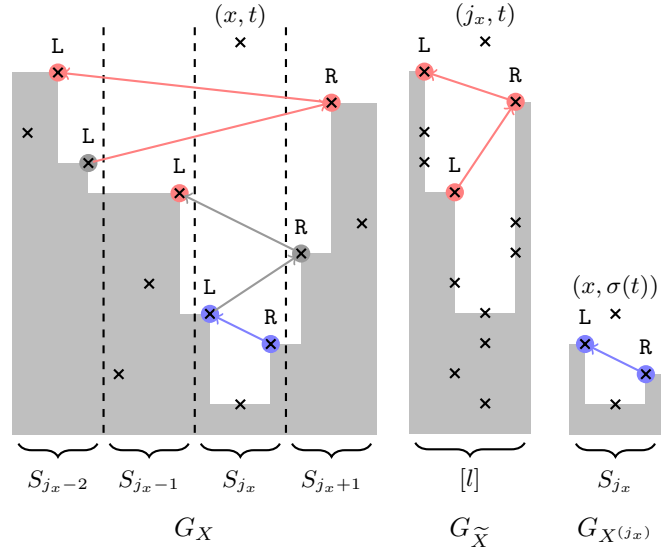

%% file: appendix.tex
\input{upper}

%% file: upper.tex
\section{Proof of Lemma~\ref{lem:bootstrapping tango}} \label{sec:appendix:boosting tango}
In this section, we prove  \Cref{lem:bootstrapping tango}. The intuition is we essentially use Tango-trees but the reference tree is split into levels according to the parameter.
Our algorithm is based on bootstrapping the Tango tree of \cite{DHIP07}. We remark that the main difference is that our algorithm is defined by a parameter $k \in [1,\log n]$ (where we use the balanced BST to implement the split/join operations in the preferred paths of Tango trees) whereas the (properly implemented) Tango tree of \cite{DHIP07} is oblivious to the parameter. 

Consider any fixed access sequence $X = (X_1, \ldots, X_m) \in S^m$. W.l.o.g.~assume $n = |S|$ is a power of $2$, and we associate each key in $S$ with an integer in $[n]$. In this section, we will consider a fixed parameter $k$ that is an integer in $[1, \log n]$.

\begin{definition}[perfect binary tree and subtrees]
Let $\cT$ denote the perfect binary tree on $[n]$.\footnote{Since each node of $\cT$ uniquely stores a key in $[n]$, with an abuse of notation, we will use $x$ to denote both a key in $[n]$ and a node in the tree $\cT$.} For each key $x \in [n]$, define $h_x$ to be the depth of the node $x$ in $\cT$ (root has depth $1$).

Divide the tree $\cT$ into $\log n / k$ levels such that the subtrees on each level have height $k$. 

For each key $x \in [n]$, consider the path from the root to the node $x$, and we denote the subtrees encountered along this path as $\cT^{(x)}_1, \cdots, \cT^{(x)}_{\lceil h_x/k \rceil}$.
\end{definition}

Note that each tree $\cT^{(x)}_\ell$ has height $k$.

\begin{definition}[preferred child and preferred path \cite{DHIKP09}]
At any time $i$, for any node $x$ in $\cT$, we define its preferred child to be the left (or right) child if the last access to a node within $x$'s subtree was in the left (or right) subtree of $x$.

Each subtree of $\cT$ can be decomposed into disjoint sets of preferred paths by recursively following the preferred child.
\end{definition}

\begin{definition}[number of changed preferred children]
At any time $i$, and for any $\ell \in [\lceil h_{X_i}/k \rceil]$, define $a^{(X_i)}_{\ell}$ to be the number of nodes in $\cT^{(X_i)}_{\ell}$ whose preferred child changes during access of $X_i$. Define $a^{(X_i)} = \sum_{\ell=1}^{\lceil h_{X_i}/k \rceil} a^{(X_i)}_{\ell}$ to be the number of all nodes in the whole tree whose preferred child changes during access of $X_i$.
\end{definition}

By the definition of the alternation bound in \Cref{def:ib}, we directly have the following lemma.
\begin{lemma}[alternation bound]
$\ALT_{\cT}(X) = \sum_{i=1}^m a^{(X_i)}$.
\end{lemma}

Similar to the original tango tree in \cite{DHIKP09}, we maintain each preferred path as an auxiliary tree. Since we define a preferred path to be a path \emph{within a subtree}, in the auxiliary tree we store the depth of a node \emph{with respect to that subtree} instead of its depth with respect to the whole tree $\cT$. We use the following result from  \cite{DHIKP09}:
\begin{lemma}[auxiliary tree, Section 3.2 of \cite{DHIKP09}]
The auxiliary tree data structure is an augmented BST that stores a preferred path and supports ``search'', ``cut by depth'', and ``join'' operations. Each operation takes $O(\log p + 1)$ time, where $p$ is the total number of nodes in the preferred path.
\end{lemma}

\paragraph*{Algorithm.} Now we are ready to describe our algorithm. Our data structure is an augmented BST that consists of all the auxiliary trees of all the preferred paths.

When accessing $X_i$, the algorithm follows the BST until it reaches $X_i$. Whenever a preferred child is changed, we use the ``cut by depth'', and ``join'' operations to change the auxiliary trees accordingly.

\begin{lemma}[time of one access]
The time to search for $X_i$ is 
\[
O\Big( (a^{(X_i)} + \frac{\log n}{k}) \cdot (\log k +1 ) \Big).
\]
\end{lemma}
\begin{proof}
We first compute the time that the algorithm spends within a subtree $\cT^{(X_i)}_{\ell}$ for any $\ell \in [\lceil h_{X_i}/k\rceil]$. Since $\cT^{(X_i)}_{\ell}$ has height $k$, the cost of each operation of an auxiliary tree is $O(\log k + 1)$. The algorithm needs to make at least one search operation in one auxiliary tree, and apart from this, the algorithm also needs to make $O(1)$ operations whenever a preferred child is changed, so the search time of $X_i$ within the subtree $\cT^{(X_i)}_{\ell}$ is 
\[
O\Big((a^{(X_i)}_{\ell} + 1) (1 + \log k)\Big).
\]

Thus the total search time of $X_i$ is
\[
\sum_{\ell=1}^{\lceil h_{X_i}/k\rceil} O\Big((a^{(X_i)}_{\ell} + 1) (1 + \log k)\Big) = O\Big( (a^{(X_i)} + \frac{\log n}{k}) \cdot (\log k +1 ) \Big),
\]
where we used the fact that $ h_{X_i} \leq \log n$ and $a^{(X_i)} = \sum_{\ell=1}^{\lceil h_{X_i}/k \rceil} a^{(X_i)}_{\ell}$.
\end{proof}

Since $\ALT_{\cT}(X) = \sum_{i=1}^m a^{(X_i)}$, we can easily compute the total access time.
\begin{lemma}[total time]
The total access time of the sequence $X = (X_1, \cdots, X_m)$ is
\[
O\Big( (\ALT_{\cT}(X) + m \cdot \frac{\log n}{k}) \cdot (\log k + 1) \Big).
\]
\end{lemma}

%% file: main.bbl
\begin{thebibliography}{10}

\bibitem{AM78}
Brian Allen and Ian Munro.
\newblock Self-organizing binary search trees.
\newblock {\em J. ACM}, 25(4):526--535, oct 1978.
\newblock URL: \url{http://doi.acm.org/10.1145/322092.322094}, \href {https://doi.org/10.1145/322092.322094} {\path{doi:10.1145/322092.322094}}.

\bibitem{BBCR09}
Boaz Barak, Mark Braverman, Xi~Chen, and Anup Rao.
\newblock Direct sums in randomized communication complexity.
\newblock {\em Electron. Colloquium Comput. Complex.}, {TR09-044}, 2009.
\newblock URL: \url{https://eccc.weizmann.ac.il/report/2009/044}, \href {https://arxiv.org/abs/TR09-044} {\path{arXiv:TR09-044}}.

\bibitem{OptPattern23}
Benjamin~Aram Berendsohn, L{\'{a}}szl{\'{o}} Kozma, and Michal Opler.
\newblock Optimization with pattern-avoiding input.
\newblock {\em CoRR}, abs/2310.04236, 2023.
\newblock \href {https://doi.org/10.48550/arXiv.2310.04236} {\path{doi:10.48550/arXiv.2310.04236}}.

\bibitem{BlellochD23}
Guy~E. Blelloch and Magdalen Dobson.
\newblock The geometry of tree-based sorting.
\newblock In {\em {ICALP}}, volume 261 of {\em LIPIcs}, pages 26:1--26:19. Schloss Dagstuhl - Leibniz-Zentrum f{\"{u}}r Informatik, 2023.
\newblock \href {https://doi.org/10.4230/LIPIcs.ICALP.2023.26} {\path{doi:10.4230/LIPIcs.ICALP.2023.26}}.

\bibitem{BDDF10}
Prosenjit Bose, Karim Dou{\"{\i}}eb, Vida Dujmovic, and Rolf Fagerberg.
\newblock An \emph{O}(log log \emph{n})-competitive binary search tree with optimal worst-case access times.
\newblock In {\em Algorithm Theory - {SWAT} 2010, 12th Scandinavian Symposium and Workshops on Algorithm Theory, Bergen, Norway, June 21-23, 2010. Proceedings}, pages 38--49, 2010.
\newblock \href {https://doi.org/10.1007/978-3-642-13731-0\_5} {\path{doi:10.1007/978-3-642-13731-0\_5}}.

\bibitem{ChalermsookCS20}
Parinya Chalermsook, Julia Chuzhoy, and Thatchaphol Saranurak.
\newblock Pinning down the strong wilber 1 bound for binary search trees.
\newblock In {\em {APPROX-RANDOM}}, volume 176 of {\em LIPIcs}, pages 33:1--33:21. Schloss Dagstuhl - Leibniz-Zentrum f{\"{u}}r Informatik, 2020.
\newblock \href {https://doi.org/10.4230/LIPIcs.APPROX/RANDOM.2020.33} {\path{doi:10.4230/LIPIcs.APPROX/RANDOM.2020.33}}.

\bibitem{ChalermsookG0MS15}
Parinya Chalermsook, Mayank Goswami, L{\'{a}}szl{\'{o}} Kozma, Kurt Mehlhorn, and Thatchaphol Saranurak.
\newblock Greedy is an almost optimal deque.
\newblock In {\em {WADS}}, volume 9214 of {\em Lecture Notes in Computer Science}, pages 152--165. Springer, 2015.
\newblock \href {https://doi.org/10.1007/978-3-319-21840-3_13} {\path{doi:10.1007/978-3-319-21840-3_13}}.

\bibitem{CGKMS15}
Parinya Chalermsook, Mayank Goswami, L{\'{a}}szl{\'{o}} Kozma, Kurt Mehlhorn, and Thatchaphol Saranurak.
\newblock Pattern-avoiding access in binary search trees.
\newblock In Venkatesan Guruswami, editor, {\em {IEEE} 56th Annual Symposium on Foundations of Computer Science, {FOCS} 2015, Berkeley, CA, USA, 17-20 October, 2015}, pages 410--423. {IEEE} Computer Society, 2015.
\newblock \href {https://doi.org/10.1109/FOCS.2015.32} {\path{doi:10.1109/FOCS.2015.32}}.

\bibitem{chalermsook2016landscape}
Parinya Chalermsook, Mayank Goswami, L{\'a}szl{\'o} Kozma, Kurt Mehlhorn, and Thatchaphol Saranurak.
\newblock The landscape of bounds for binary search trees.
\newblock {\em arXiv preprint arXiv:1603.04892}, 2016.

\bibitem{ChalermsookGJAP23}
Parinya Chalermsook, Manoj Gupta, Wanchote Jiamjitrak, Nidia~Obscura Acosta, Akash Pareek, and Sorrachai Yingchareonthawornchai.
\newblock Improved pattern-avoidance bounds for greedy bsts via matrix decomposition.
\newblock In {\em {SODA}}, pages 509--534. {SIAM}, 2023.
\newblock \href {https://doi.org/10.1137/1.9781611977554.ch22} {\path{doi:10.1137/1.9781611977554.ch22}}.

\bibitem{CGJPY23}
Parinya Chalermsook, Manoj Gupta, Wanchote Jiamjitrak, Akash Pareek, and Sorrachai Yingchareonthawornchai.
\newblock The group access bounds for binary search trees.
\newblock {\em CoRR}, abs/2312.15426, 2023.
\newblock \href {https://doi.org/10.48550/arXiv.2312.15426} {\path{doi:10.48550/arXiv.2312.15426}}.

\bibitem{ChalermsookJ20}
Parinya Chalermsook and Wanchote~Po Jiamjitrak.
\newblock New binary search tree bounds via geometric inversions.
\newblock In {\em {ESA}}, volume 173 of {\em LIPIcs}, pages 28:1--28:16. Schloss Dagstuhl - Leibniz-Zentrum f{\"{u}}r Informatik, 2020.
\newblock \href {https://doi.org/10.4230/LIPIcs.ESA.2020.28} {\path{doi:10.4230/LIPIcs.ESA.2020.28}}.

\bibitem{CPY24}
Parinya Chalermsook, Seth Pettie, and Sorrachai Yingchareonthawornchai.
\newblock Sorting pattern-avoiding permutations via 0-1 matrices forbidding product patterns.
\newblock In {\em {SODA}}, pages 133--149. {SIAM}, 2024.
\newblock \href {https://doi.org/10.1137/1.9781611977912.7} {\path{doi:10.1137/1.9781611977912.7}}.

\bibitem{Cole00}
Richard Cole.
\newblock On the dynamic finger conjecture for splay trees. part {II:} the proof.
\newblock {\em {SIAM} J. Comput.}, 30(1):44--85, 2000.
\newblock \href {https://doi.org/10.1137/S009753979732699X} {\path{doi:10.1137/S009753979732699X}}.

\bibitem{DHIKP09}
Erik~D. Demaine, Dion Harmon, John Iacono, Daniel~M. Kane, and Mihai Patrascu.
\newblock The geometry of binary search trees.
\newblock In {\em Proceedings of the Twentieth Annual {ACM-SIAM} Symposium on Discrete Algorithms, {SODA} 2009, New York, NY, USA, January 4-6, 2009}, pages 496--505, 2009.
\newblock URL: \url{http://dl.acm.org/citation.cfm?id=1496770.1496825}, \href {https://doi.org/10.1137/1.9781611973068.55} {\path{doi:10.1137/1.9781611973068.55}}.

\bibitem{DHIP07}
Erik~D. Demaine, Dion Harmon, John Iacono, and Mihai Patrascu.
\newblock Dynamic optimality - almost.
\newblock {\em {SIAM} J. Comput.}, 37(1):240--251, 2007.
\newblock \href {https://doi.org/10.1137/S0097539705447347} {\path{doi:10.1137/S0097539705447347}}.

\bibitem{DerrberrySW05}
Jonathan Derryberry, Daniel~Dominic Sleator, and Chengwen~Chris Wang.
\newblock A lower bound framework for binary search trees with rotations.
\newblock {\em Technical report}, 2005.

\bibitem{Elmasry04}
Amr Elmasry.
\newblock On the sequential access theorem and deque conjecture for splay trees.
\newblock {\em Theor. Comput. Sci.}, 314(3):459--466, 2004.
\newblock \href {https://doi.org/10.1016/j.tcs.2004.01.019} {\path{doi:10.1016/j.tcs.2004.01.019}}.

\bibitem{Fox11}
Kyle Fox.
\newblock Upper bounds for maximally greedy binary search trees.
\newblock In {\em {WADS}}, volume 6844 of {\em Lecture Notes in Computer Science}, pages 411--422. Springer, 2011.
\newblock \href {https://doi.org/10.1007/978-3-642-22300-6_35} {\path{doi:10.1007/978-3-642-22300-6_35}}.

\bibitem{GNW11}
Oded Goldreich, Noam Nisan, and Avi Wigderson.
\newblock On yao's xor-lemma.
\newblock In Oded Goldreich, editor, {\em Studies in Complexity and Cryptography. Miscellanea on the Interplay between Randomness and Computation - In Collaboration with Lidor Avigad, Mihir Bellare, Zvika Brakerski, Shafi Goldwasser, Shai Halevi, Tali Kaufman, Leonid Levin, Noam Nisan, Dana Ron, Madhu Sudan, Luca Trevisan, Salil Vadhan, Avi Wigderson, David Zuckerman}, volume 6650 of {\em Lecture Notes in Computer Science}, pages 273--301. Springer, 2011.
\newblock \href {https://doi.org/10.1007/978-3-642-22670-0\_23} {\path{doi:10.1007/978-3-642-22670-0\_23}}.

\bibitem{GoyalG19}
Navin Goyal and Manoj Gupta.
\newblock Better analysis of greedy binary search tree on decomposable sequences.
\newblock {\em Theor. Comput. Sci.}, 776:19--42, 2019.
\newblock \href {https://doi.org/10.1016/j.tcs.2018.12.021} {\path{doi:10.1016/j.tcs.2018.12.021}}.

\bibitem{IaconoSurvey}
John Iacono.
\newblock In pursuit of the dynamic optimality conjecture.
\newblock In {\em Space-Efficient Data Structures, Streams, and Algorithms - Papers in Honor of J. Ian Munro on the Occasion of His 66th Birthday}, pages 236--250, 2013.
\newblock \href {https://doi.org/10.1007/978-3-642-40273-9\_16} {\path{doi:10.1007/978-3-642-40273-9\_16}}.

\bibitem{KRW95}
Mauricio Karchmer, Ran Raz, and Avi Wigderson.
\newblock Super-logarithmic depth lower bounds via the direct sum in communication complexity.
\newblock {\em Comput. Complex.}, 5(3/4):191--204, 1995.
\newblock \href {https://doi.org/10.1007/BF01206317} {\path{doi:10.1007/BF01206317}}.

\bibitem{KLM021}
Alexander Knop, Shachar Lovett, Sam McGuire, and Weiqiang Yuan.
\newblock Log-rank and lifting for and-functions.
\newblock In Samir Khuller and Virginia~Vassilevska Williams, editors, {\em {STOC} '21: 53rd Annual {ACM} {SIGACT} Symposium on Theory of Computing, Virtual Event, Italy, June 21-25, 2021}, pages 197--208. {ACM}, 2021.
\newblock \href {https://doi.org/10.1145/3406325.3450999} {\path{doi:10.1145/3406325.3450999}}.

\bibitem{kozma2016binary}
L{\'a}szl{\'o} Kozma.
\newblock Binary search trees, rectangles and patterns.
\newblock 2016.

\bibitem{KS18}
L{\'{a}}szl{\'{o}} Kozma and Thatchaphol Saranurak.
\newblock Smooth heaps and a dual view of self-adjusting data structures.
\newblock In {\em Proceedings of the 50th Annual {ACM} {SIGACT} Symposium on Theory of Computing, {STOC} 2018, Los Angeles, CA, USA, June 25-29, 2018}, pages 801--814, 2018.
\newblock \href {https://doi.org/10.1145/3188745.3188864} {\path{doi:10.1145/3188745.3188864}}.

\bibitem{lecomte2019settling}
Victor Lecomte and Omri Weinstein.
\newblock Settling the relationship between wilber's bounds for dynamic optimality.
\newblock In {\em {ESA}}, volume 173 of {\em LIPIcs}, pages 68:1--68:21. Schloss Dagstuhl - Leibniz-Zentrum f{\"{u}}r Informatik, 2020.
\newblock \href {https://doi.org/10.4230/LIPIcs.ESA.2020.68} {\path{doi:10.4230/LIPIcs.ESA.2020.68}}.

\bibitem{LSS08}
Troy Lee, Adi Shraibman, and Robert Spalek.
\newblock A direct product theorem for discrepancy.
\newblock In {\em {CCC}}, pages 71--80. {IEEE} Computer Society, 2008.
\newblock \href {https://doi.org/10.1109/CCC.2008.25} {\path{doi:10.1109/CCC.2008.25}}.

\bibitem{LT19}
Caleb~C. Levy and Robert~E. Tarjan.
\newblock A new path from splay to dynamic optimality.
\newblock In Timothy~M. Chan, editor, {\em Proceedings of the Thirtieth Annual {ACM-SIAM} Symposium on Discrete Algorithms, {SODA} 2019, San Diego, California, USA, January 6-9, 2019}, pages 1311--1330. {SIAM}, 2019.
\newblock \href {https://doi.org/10.1137/1.9781611975482.80} {\path{doi:10.1137/1.9781611975482.80}}.

\bibitem{Luc88}
Joan~Marie Lucas.
\newblock {\em Canonical forms for competitive binary search tree algorithms}.
\newblock Rutgers University, Department of Computer Science, Laboratory for Computer Science Research, 1988.

\bibitem{Munro00}
J.~Ian Munro.
\newblock On the competitiveness of linear search.
\newblock In Mike Paterson, editor, {\em Algorithms - {ESA} 2000, 8th Annual European Symposium, Saarbr{\"{u}}cken, Germany, September 5-8, 2000, Proceedings}, volume 1879 of {\em Lecture Notes in Computer Science}, pages 338--345. Springer, 2000.
\newblock \href {https://doi.org/10.1007/3-540-45253-2\_31} {\path{doi:10.1007/3-540-45253-2\_31}}.

\bibitem{pettie2007splay}
Seth Pettie.
\newblock Splay trees, davenport-schinzel sequences, and the deque conjecture.
\newblock In {\em Proceedings of the nineteenth annual ACM-SIAM symposium on Discrete algorithms}, pages 1115--1124, 2008.
\newblock URL: \url{http://dl.acm.org/citation.cfm?id=1347082.1347204}.

\bibitem{PT09}
Mihai P{u{a}}tra{c{s}}cu and Mikkel Thorup.
\newblock Higher lower bounds for near-neighbor and further rich problems.
\newblock {\em {SIAM} J. Comput.}, 39(2):730--741, 2009.
\newblock \href {https://doi.org/10.1137/070684859} {\path{doi:10.1137/070684859}}.

\bibitem{Raz11}
Ran Raz.
\newblock A counterexample to strong parallel repetition.
\newblock {\em {SIAM} J. Comput.}, 40(3):771--777, 2011.
\newblock \href {https://doi.org/10.1137/090747270} {\path{doi:10.1137/090747270}}.

\bibitem{RM99}
Ran Raz and Pierre McKenzie.
\newblock Separation of the monotone {NC} hierarchy.
\newblock {\em Comb.}, 19(3):403--435, 1999.
\newblock \href {https://doi.org/10.1007/s004930050062} {\path{doi:10.1007/s004930050062}}.

\bibitem{SadehK23}
Yaniv Sadeh and Haim Kaplan.
\newblock Dynamic binary search trees: Improved lower bounds for the greedy-future algorithm.
\newblock In {\em {STACS}}, volume 254 of {\em LIPIcs}, pages 53:1--53:21. Schloss Dagstuhl - Leibniz-Zentrum f{\"{u}}r Informatik, 2023.
\newblock \href {https://doi.org/10.4230/LIPIcs.STACS.2023.53} {\path{doi:10.4230/LIPIcs.STACS.2023.53}}.

\bibitem{Sha01_counterex_DP}
Ronen Shaltiel.
\newblock Towards proving strong direct product theorems.
\newblock In {\em Proceedings of the 16th Annual Conference on Computational Complexity}, CCC '01, page 107, USA, 2001. IEEE Computer Society.

\bibitem{Sherstov12}
Alexander~A. Sherstov.
\newblock Strong direct product theorems for quantum communication and query complexity.
\newblock {\em {SIAM} J. Comput.}, 41(5):1122--1165, 2012.
\newblock \href {https://doi.org/10.1137/110842661} {\path{doi:10.1137/110842661}}.

\bibitem{Shitov18}
Yaroslav Shitov.
\newblock A counterexample to comon's conjecture.
\newblock {\em {SIAM} J. Appl. Algebra Geom.}, 2(3):428--443, 2018.
\newblock \href {https://doi.org/10.1137/17M1131970} {\path{doi:10.1137/17M1131970}}.

\bibitem{ST85}
Daniel~Dominic Sleator and Robert~Endre Tarjan.
\newblock Self-adjusting binary search trees.
\newblock {\em J. ACM}, 32(3):652--686, jul 1985.
\newblock URL: \url{http://doi.acm.org/10.1145/3828.3835}, \href {https://doi.org/10.1145/3828.3835} {\path{doi:10.1145/3828.3835}}.

\bibitem{Sundar92}
Rajamani Sundar.
\newblock On the deque conjecture for the splay algorithm.
\newblock {\em Comb.}, 12(1):95--124, 1992.
\newblock \href {https://doi.org/10.1007/BF01191208} {\path{doi:10.1007/BF01191208}}.

\bibitem{Multisplay}
Chengwen~Chris Wang, Jonathan Derryberry, and Daniel~Dominic Sleator.
\newblock O(log log n)-competitive dynamic binary search trees.
\newblock In {\em Proceedings of the Seventeenth Annual ACM-SIAM Symposium on Discrete Algorithm}, SODA '06, pages 374--383, Philadelphia, PA, USA, 2006. Society for Industrial and Applied Mathematics.
\newblock URL: \url{http://dl.acm.org/citation.cfm?id=1109557.1109600}.

\bibitem{Wil89}
R.~Wilber.
\newblock Lower bounds for accessing binary search trees with rotations.
\newblock {\em SIAM J. Comput.}, 18(1):56--67, feb 1989.
\newblock URL: \url{http://dx.doi.org/10.1137/0218004}, \href {https://doi.org/10.1137/0218004} {\path{doi:10.1137/0218004}}.

\end{thebibliography}
